\documentclass[11pt,a4paper]{article}
\usepackage[a4paper]{geometry}
\usepackage{amssymb,latexsym,amsmath,amsfonts,amsthm}
\usepackage{fullpage}
\usepackage{comment}
\usepackage{overpic}
\usepackage{mathrsfs,enumerate}
\usepackage{amsmath}
\usepackage{amsthm}
\usepackage{amssymb}
\usepackage{amsfonts}
\usepackage{epsf}
\usepackage{hyperref}
\usepackage{indentfirst}
\usepackage{color}

\renewcommand{\Re}{\mathrm{Re}\,}
\renewcommand{\Im}{\mathrm{Im}\,}
\newcommand{\ds}{\displaystyle}

\renewcommand{\vec}{\mathbf}

\newtheorem{thm}{Theorem}[section]

\newtheorem{lem}[thm]{Lemma}
\newtheorem{prop}[thm]{Proposition}

\newtheorem{rhp}[thm]{RH problem}
\numberwithin{equation}{section}

\newcommand{\al}{\alpha}

\theoremstyle{remark}
\newtheorem{remark}[thm]{Remark}
\numberwithin{equation}{section}

\newcommand{\eq}{\begin{equation}}
\newcommand{\nq}{\end{equation}}
\newcommand{\eqa}{\begin{eqnarray}}
\newcommand{\nqa}{\end{eqnarray}}

\begin{document}

\title{Gap probability at the hard edge for random matrix ensembles with pole singularities in the potential}
\author{Dan Dai\footnotemark[1], ~Shuai-Xia Xu\footnotemark[2] ~and Lun Zhang\footnotemark[3]}

\renewcommand{\thefootnote}{\fnsymbol{footnote}}
\footnotetext[1]{Department of Mathematics, City University of Hong Kong, Tat Chee
Avenue, Kowloon, Hong Kong. E-mail: \texttt{dandai@cityu.edu.hk}}
\footnotetext[2]{Institut Franco-Chinois de l'Energie Nucl\'{e}aire, Sun Yat-sen University,
Guangzhou 510275, China. E-mail: \texttt{xushx3@mail.sysu.edu.cn}}
\footnotetext[3] {School of Mathematical Sciences and Shanghai Key Laboratory for Contemporary Applied Mathematics, Fudan University, Shanghai 200433, China. E-mail: \texttt{lunzhang@fudan.edu.cn }}
\maketitle

\begin{abstract}
We study the Fredholm determinant of an integrable operator acting on the interval $(0,s)$ whose kernel is constructed out of the $\Psi$-function associated with a hierarchy of higher order analogues to the Painlev\'{e} III equation. This Fredholm determinant describes the critical behavior of the eigenvalue gap probability at the hard edge of unitary invariant random matrix ensembles perturbed
by poles of order $k$ in a certain scaling regime. Using the Riemann-Hilbert method, we obtain the large $s$ asymptotics of the Fredholm determinant. Moreover, we derive a Painlev\'e type formula of the Fredholm determinant, which is expressed in terms of an explicit integral involving a solution to a coupled Painlev\'e III system.
\end{abstract}

\medskip


\tableofcontents
\section{Introduction}
\subsection{Random matrix ensembles with pole singularities in the potential}
In this paper, we are concerned with the following unitary invariant random matrix ensembles
\begin{equation}\label{eq:probmeasure}
\frac{1}{Z_n} (\det M)^\alpha \exp[- n\textrm{tr}\, V_k(M)]dM, \qquad \alpha>-1,
\end{equation}
defined on the space $\mathcal{H}_n^+$ of $n \times n$ positive definite Hermitian matrices $M=(M_{ij})_{1\leq i,j \leq n}$, where
\begin{equation}
dM=\prod_{i=1}^{n}dM_{ii} \prod_{1 \leq i < j \leq n} d\Re M_{ij}d\Im M_{ij},
\end{equation}
\begin{equation}
  Z_n=\int_{\mathcal{H}_n^+} (\det M)^\alpha \exp[- n\textrm{tr}\, V_k(M)]dM
\end{equation}
is the normalization constant,  and the potential
\begin{equation} \label{vt-def}
  V_k(x):= V(x) + \left(\frac{t}{x}\right)^k, \qquad x \in (0,\infty), \quad t >0 ,  \quad k \in \mathbb{N}.
\end{equation}
In \eqref{vt-def}, it is assumed that the regular part $V$ of the potential is real analytic on $[0, \infty)$, and also we are in the one-cut regular case in the sense of \cite{DKMVZ99,KM2000}. This particularly implies that the limiting mean distribution of the eigenvalues as $n\to\infty$ for $t=0$ is supported on $[0,b]$ for some positive $b$.

Since the ensembles are unitary invariant, we have that the eigenvalues of a random matrix in \eqref{eq:probmeasure} form a determinantal point process \cite{Johansson06,Soshnikov00}, whose joint probability density function is explicitly given by
\begin{equation}\label{eq:jpdf}
\frac{1}{\widehat{Z}_n}\prod_{1\leq i < j \leq n}(x_j-x_i)^2\prod_{j=1}^n w(x_j),
\end{equation}
with
$$\widehat{Z}_n=\int_{[0,\infty)^n}\prod_{1\leq i < j \leq n}(x_j-x_i)^2\prod_{j=1}^n w(x_j)dx_j$$
and
\begin{equation}\label{eq:weight}
w(x)=x^\alpha e^{-nV_k(x)}.
\end{equation}

Clearly, random matrix ensembles of the form \eqref{eq:probmeasure} can be interpreted as unitary ensembles perturbed by a pole of order $k$ at the origin. These singularly perturbed ensembles have attracted lots of interests recently. On one hand, they arise in a variety of problems. These problems include statistics for zeta zeros and eigenvalues \cite{BS08}, eigenvalues of Wigner-Smith time-delay matrix in the context of quantum transport and electrical characteristics of chaotic cavities \cite{Brouwer99,Brouwer97,TM13}, and random matrix models in the field of spin-glasses \cite{Akemann14}, etc.. On the other hand, these ensembles are natural candidates to exhibit new critical phenomena. Indeed, for each fixed $t>0$, the eigenvalues are repelled from the origin due to the pole singularities in the potential. This behavior is quite different from the unperturbed case (i.e., $t=0$), in which the eigenvalues will accumulate near the origin as $n \to \infty$ for the one-cut regular case. As a consequence, it is expected that some new phenomena will appear near the origin if one takes $t\to 0$ and simultaneously $n\to \infty$, which also corresponds to a phase transition between two different edge behaviors. Studies of this aspect have been conducted in \cite{Bri:Mezz:Mo} for the perturbed Gaussian unitary ensembles, in \cite{ChenIts2010,Lyu:Gri:Chen2017,XDZ2015,XDZ} for the perturbed Laguerre unitary ensembles, and in \cite{ACM} for the general case as indicated in \eqref{eq:probmeasure}. It comes out that the Painlev\'e III equation and its hierarchy play an important role in describing the critical behavior, as we will review from the viewpoint of correlation kernel in what follows.


\subsection{Double scaling limit of the correlation kernel at the hard edge}
Let us denote by $K_{n}(x,y;t)$ the correlation kernel associated with \eqref{eq:jpdf}, which takes the following form
\begin{equation} \label{eq:correlation kernel}
  K_{n}(x,y;t)=h_{n-1}^{-1} \sqrt{w(x) w(y)}
\frac {\pi_n(x)\pi_{n-1}(y)-\pi_{n-1}(x)\pi_n(y)}{x-y}, \quad x,y >0.
\end{equation}
Here $\pi_j(x)$, $j=0,1,\cdots,$ is a family of monic polynomials of degree $j$ satisfying the orthogonality conditions
\begin{equation}
\int_0^{\infty}\pi_j(x)\pi_m(x)w(x)dx=h_j\delta_{j,m},
\end{equation}
where the weight function $w$ is given in \eqref{eq:weight}.

To state the relevant results, we also need the following Riemann-Hilbert(RH) problem.
\begin{figure}[h]
\begin{center}
   \setlength{\unitlength}{1truemm}
   \begin{picture}(100,70)(-5,2)
       \put(40,40){\line(-1,-1){25}}
       \put(40,40){\line(-1,1){25}}
       \put(40,40){\line(-1,0){30}}

       \put(25,55){\thicklines\vector(1,-1){1}}
       \put(30,40){\thicklines\vector(1,0){1}}
       \put(25,25){\thicklines\vector(1,1){1}}

       \put(39,36.3){$0$}
       \put(20,11){$\Sigma_3$}
       \put(20,66){$\Sigma_1$}
       \put(3,40){$\Sigma_2$}
       \put(25,44){$\Omega_2$}
       \put(25,34){$\Omega_3$}
       \put(55,40){$\Omega_1$}

       \put(40,40){\thicklines\circle*{1}}

   \end{picture}
   \caption{The jump contours $\Sigma_j$ and the regions $\Omega_j$, $j=1,2,3,$ for the RH problem for $\Psi$.}
   \label{fig:jumpsPsi}
\end{center}
\end{figure}
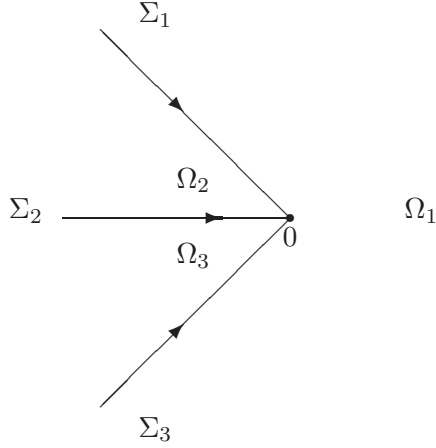
\begin{rhp} \label{rhp:Psi}
Let $\lambda>0$ and $\alpha > -1$. We look for a $2\times 2$ matrix-valued function
$\Psi(z)=\Psi(z;\lambda) $
satisfying
\begin{enumerate}
\item[\rm (1)] $\Psi(z)$ is defined and analytic in $\mathbb{C}\setminus \{\cup^3_{j=1}\Sigma_j\cup\{0\}\}$, where the contours $\Sigma_j$, $j=1,2,3$,  are illustrated in Figure \ref{fig:jumpsPsi}.
\item[\rm (2)] $\Psi$ has limiting values $\Psi_{\pm}(z)$ for $z\in \cup^3_{j=1}\Sigma_j$, where $\Psi_+$ ($\Psi_-$) denotes the limiting values from the left (right) side of $\Sigma_j$, and
\begin{equation}\label{eq:Psi-jump}
 \Psi_+(z)=\Psi_-(z)
 \left\{
 \begin{array}{ll}
    \left(
                               \begin{array}{cc}
                                 1 & 0 \\
                               e^{\alpha \pi i}& 1 \\
                                 \end{array}
                             \right), &  z \in \Sigma_1, \\[.4cm]
    \left(
                               \begin{array}{cc}
                                0 & 1\\
                               -1 & 0 \\
                                 \end{array}
                             \right),  &  z \in \Sigma_2, \\[.4cm]
    \left(
                               \begin{array}{cc}
                                 1 & 0 \\
                                 e^{-\alpha \pi i} & 1 \\
                                 \end{array}
                             \right), &   z \in \Sigma_3.
 \end{array}  \right .  \end{equation}

\item[\rm (3)] As $z\to\infty$, there exist functions $p(\lambda)$, $q(\lambda)$ and $r(\lambda)$ such that
\begin{equation}\label{eq:Psi-infty}
 \Psi(z;\lambda)=
 \left( I + \frac{1}{z}
 \begin{pmatrix}
 q(\lambda) & -i r(\lambda) \\
 i p(\lambda) & -q(\lambda)
 \end{pmatrix} +
O(z^{-2}) \right) z^{-\frac{1}{4} \sigma_3} \frac{I + i \sigma_1}{\sqrt{2}} e^{z^{\frac{1}{2}} \sigma_3},
   \end{equation}
   where the branch cuts of $z^{-\frac{1}{4}}$ and $z^{\frac{1}{2}}$ are taken along the negative real axis, and
$$\sigma_1=\begin{pmatrix}
 0 & 1 \\
 1 & 0
 \end{pmatrix}, \qquad \sigma_3=\begin{pmatrix}
 1 & 0 \\
 0 & -1
 \end{pmatrix}
 $$
are the Pauli matrices.
\item[\rm (4)]As $z\to 0$, there exists a unimodular matrix $\Psi_0(\lambda)$, independent of $z$, such that
\begin{equation}\label{eq:Psi-origin}
\Psi(z;\lambda)=\Psi_0(\lambda)(I+O(z)) e^{-\left(-\frac{\lambda}{z}\right)^k\sigma_3}z^{\frac \alpha2\sigma_3}\left\{
  \begin{array}{ll}
    I, &  z \in \Omega_1, \\
   \left(
                               \begin{array}{cc}
                                 1 & 0 \\
                               -e^{\alpha\pi i}& 1 \\
                                 \end{array}
                             \right),  &  z \in \Omega_2, \\
    \left(
                               \begin{array}{cc}
                                 1 &0 \\
                                  e^{-\alpha \pi i} &1 \\
                                 \end{array}
                             \right), &   z \in \Omega_3,
 \end{array}  \right .
  \end{equation}
where the regions $\Omega_i$, $i=1,2,3$, are depicted in Figure \ref{fig:jumpsPsi}.
\end{enumerate}
\end{rhp}

It was shown in \cite{ACM} that there exists a unique solution to the above RH problem for each $k\in\mathbb{N}$, $\alpha>-1$ and $\lambda>0$.  This generalizes the result in \cite{XDZ} for the case $k=1$.

Now, for $z\in\Omega_1$, let us set
\begin{equation}
\begin{pmatrix}
\psi_1(z;\lambda) \\
\psi_2(z;\lambda)
\end{pmatrix}:=\Psi(z;\lambda)\begin{pmatrix}
1 \\
0
\end{pmatrix}.
\end{equation}
These two functions can be extended analytically in the whole complex plane with a cut along $\Sigma_1$, which we still denote by $\psi_i(z;\lambda)$, $i=1,2$. This particularly implies that
\begin{equation}\label{eq:psinegx}
\begin{pmatrix}
\psi_1(x;\lambda) \\
\psi_2(x;\lambda)
\end{pmatrix}=\Psi_-(x;\lambda)\begin{pmatrix}
1 \\
-e^{-\alpha \pi i}
\end{pmatrix}, \qquad x<0.
\end{equation}

Let us consider the correlation kernel $K_n(x,y;t)$ given in \eqref{eq:correlation kernel}. Assume that the parameter $t\to 0$ and $n\to \infty$ in such a way that
\begin{equation}\label{eq:doublescaling}
2^{-\frac{1}{k}}c_V n^{\frac{2k+1}{k}}t \to \lambda>0,
\end{equation}
where $c_V$ is a constant depending on the regular part $V$ in the potential. It was shown in \cite[Theorem 1.10]{ACM} that
\begin{equation}\label{eq:hardedgelimit}
 \lim_{n\to\infty} \frac {1}{c_Vn^2} K_n\left (\frac {u}{c_Vn^2} ,\frac {v}{c_Vn^2}; t \right) = K_{\textrm{PIII}}(u,v;\lambda),
  \end{equation}
uniformly  for $u$, $v$ and $\lambda$  in  any compact subsets of  $(0,\infty)$,
where
\begin{equation}\label{eq:PIIIkernel}
 K_{\textrm{PIII}}(u,v;\lambda):=e^{\alpha \pi i }\frac{\psi_1(-u;\lambda)\psi_2(-v;\lambda)-  \psi_1(-v;\lambda)\psi_2(-u;\lambda)}{2\pi i(u-v)}.
  \end{equation}
The subscript $\textrm{PIII}$ is used to indicate that the RH problem \ref{rhp:Psi} for $\Psi$ is related to a hierarchy of higher order analogues to the Painlev\'{e} III equation, which we will explain next.


\begin{remark} \label{remark-Bessel relation}
When $\lambda=0$, RH problem \ref{rhp:Psi} for $\Psi$ can be solved explicitly in terms of modified Bessel functions; cf. \cite{KMVV2004}. Furthermore, this also corresponds to the case $t=0$ in \eqref{vt-def}, i.e., the singular part in the potential $V_k$ vanishes. Since the hard edge scaling limit of the correlation kernel for the unperturbed unitary ensembles is given by the classical Bessel kernel \cite{Forrester93} in general (cf. \cite{Vanlessen}), it is then natural to expect some transition will occur between these two kernels as the parameter $\lambda$ varies. This is indeed the case. According to \cite{ACM,XDZ}, we have
\begin{equation}\label{eq:PIIItoBesl}
K_{\textrm{PIII}}(u,v;\lambda)=K_{\textrm{Bes}}(u,v)+O(\lambda),  \qquad \lambda \to 0^+,
\end{equation}
uniformly for $u,v$ in any compact subset of $(0,\infty)$,
where
\begin{equation}\label{def:KBes}
K_{\rm Bes}(x,y) = \frac{J_\alpha(\sqrt x)\sqrt y J'_\alpha
(\sqrt y)-\sqrt x J_\alpha'(\sqrt x)J_\alpha(\sqrt y)}{2(x-y)},\quad
\end{equation}
is the Bessel kernel of order $\alpha$, and $J_\alpha(x)$ is the Bessel function of the first kind of order
$\alpha$; see \cite{AS}.
\end{remark}

\subsection{Connection with a Painlev\'{e} III hierarchy}\label{subsec:PIIIhierarchy}
The (1,2) entry $r(\lambda)$ of the residue term at infinity for $\Psi$ in \eqref{eq:Psi-infty} is connected to a Painlev\'{e} III hierarchy. More precisely, given $k\in\mathbb{N}$, let us consider the following $k+1$ ODEs for $k+1$ unknown functions
$(\varrho(\lambda), \ell_1(\lambda), \ldots, \ell_k(\lambda))$ indexed by $p=0,\ldots, k$:
\begin{equation}\label{eq:PIIIequ}
\left\{
  \begin{array}{ll}
    \varrho=-\frac{1}{4\ell^2_k}((\ell_k^2)''-3(\ell_k')^2+\tau_0), & \hbox{$p=0$,} \\
    \sum\limits_{q=0}^p(\ell_{k-p+q+1}\ell_{k-q}-(\ell_{k-p+q}\ell_{k-q})''+3\ell'_{k-p+q}\ell'_{k-q}-4\varrho\ell_{k-p+q}\ell_{k-q})=\tau_p, & \hbox{$1\leq p \leq k$,}
  \end{array}
\right.
\end{equation}
where $(') = \frac{d}{d\lambda}$, $\tau_p$, $p=0,1,\ldots,k$ are real constants, and
$$\ell_{k+1}(\lambda)=0,\qquad \ell_0(\lambda)=\frac{\lambda}{2}.$$

If $k=1$, the ODE for $\ell_1$ reads
$$\ell_1''(\lambda)=\frac{\ell_1'(\lambda)^2}{\ell_1(\lambda)}
-\frac{\ell_1'(\lambda)}{\lambda}-\frac{\ell_1(\lambda)^2}{\lambda}
-\frac{\tau_0}{\ell_1(\lambda)}+\frac{\tau_1}{\lambda},$$
which can be identified as the Painlev\'{e} III equation \cite[Equation (32.2.3)]{DLMF} with parameters $\alpha=-1, \beta= \tau_1, \gamma = 0$ and $\delta=-\tau_0p$. Hence, the system of equations \eqref{eq:PIIIequ} is called the $k$-th member of the Painlev\'{e} III hierarchy. For general $k>1$, one can obtain an ODE of order $2k$ for $\ell_1$ from \eqref{eq:PIIIequ} by eliminating the other functions $\varrho$ and $\ell_p$ ($2\leq p \leq k$).

The connection between $r(\lambda)$ and the Painlev\'{e} III hierarchy \eqref{eq:PIIIequ} is shown in \cite[Theorem 1.6]{ACM}, as stated in the following proposition.
\begin{prop}\label{prop:PIII}
Let $\Psi(z;\lambda)$ be the unique solution to RH problem \ref{rhp:Psi} for $\alpha>-1$ and $\lambda>0$.
Then, the limit
\begin{equation}
y_{\alpha}(\lambda)=-2i\frac{d}{d\lambda}\left(\lim_{z \to \infty} z\Psi\left(z;\lambda^2\right)e^{-z^{\frac{1}{2}} \sigma_3} \frac{I - i \sigma_1}{\sqrt{2}} z^{\frac{1}{4} \sigma_3}  \right)_{12} = -2 \frac{d}{d\lambda} \left( r\left(\lambda^2\right) \right)
\end{equation}
exists,  where $(M)_{ij}$ stands for the $(i,j)$-th entry of a matrix $M$. Furthermore, $y_{\alpha}(\lambda)$ is a solution of the equation for $\ell_1$ of the $k$-th member of the Painlev\'{e} III hierarchy \eqref{eq:PIIIequ}, with the parameters $\tau_p$, $p=0,\ldots,k$ given by
\begin{equation}
\tau_p=\left\{
         \begin{array}{ll}
           4^{2k+1}k^2, & \hbox{$p=0$,} \\
          -(-4)^{k+1}\alpha k, & \hbox{$p=k$,} \\
          0, & \hbox{$0<p<k$.}
         \end{array}
       \right.
\end{equation}
\end{prop}
In addition, the asymptotic behavior of $y_{\alpha}$ can be derived from the following asymptotics of $r(\lambda)$:
\begin{equation}\label{eq:rsmall}
r(\lambda)=\frac{1}{8}\left(1-4\alpha^2\right)+O\left( \lambda^k \right), \qquad \textrm{as } \lambda \to 0^+
\end{equation}
and
\begin{equation}\label{eq:rlarge}
r(\lambda)=\lambda^{\frac{2k}{2k+1}}\left(\beta_{k-2}-\frac{3}{2}z_0\right)-(-z_0)^\frac{1}{2}\lambda^{\frac{k}{2k+1}} \alpha +O(1),
\qquad \textrm{as } \lambda \to + \infty,
\end{equation}
where
$$z_0=-\left(\frac{2^{k-1}(k-1)!}{(2k-1)!!}\right)^{-\frac{2}{2k+1}}, \quad \beta_j=-(-z_0)^{-\frac{3}{2}-j}\frac{(2j+1)!!}{2^jj!}$$
with $(2j+1)!!=(2j+1)(2j-1)\cdots 3 \cdot 1$ being the double factorial;
see \cite[Equations (5.17) and (4.57)]{ACM}.

\subsection{Gap probability at the hard edge}
 Let $\mathcal{K}_{\textrm{PIII}}$ be the integral operator with kernel $\chi_{[0,s]}(u) K_{\textrm{PIII}}(u,v) \chi_{[0,s]}(v)$ in \eqref{eq:PIIIkernel} acting on the function space $L^2((0,\infty))$, where $s>0$ and $\chi_J$ is the characteristic function of the interval $J$. It is the aim of this paper to study the Fredholm determinant
\begin{equation} \label{KPiii-def}
  \det\left(I-\mathcal{K}_{\textrm{PIII}}\right).
\end{equation}
Due to the determinantal structure \eqref{eq:jpdf}, the above function gives us the gap probability (the probability of finding no eigenvalues) on the interval $(0,s)$ for the limiting process of the random matrix ensembles \eqref{eq:probmeasure}, i.e.,
\begin{multline} \label{Kpiii-gap-prob}
\det(I-\mathcal{K}_{\textrm{PIII}})
\\
=\lim \textrm{Prob}\left(\textrm{$M$ distributed according to \eqref{eq:probmeasure} has no eigenvalues in $ \left(0, \frac{s}{c_Vn^2} \right)$}\right),
\end{multline}
where the limit is understood that both $t\to 0$ and $n\to \infty$ such that condition \eqref{eq:doublescaling} holds.

It is well-known that, for many limiting correlation kernels arising from random matrix theory, the associated Fredholm determinants are related to systems of integrable differential equations \cite{TW94c}--\cite{TW93}. Note that kernel $K_{\textrm{PIII}}(u,v)$ in \eqref{eq:PIIIkernel} can be viewed as a generalization of the classical Bessel kernel $K_{\textrm{Bes}}(u,v)$ in \eqref{def:KBes} (see Remark \ref{remark-Bessel relation}). Let us recall some basic facts regarding the Bessel kernel.

Define $\mathcal{K}_{\textrm{Bes}}$ to be the integral operator with kernel $\chi_{[0,s]}(u) K_{\textrm{Bes}}(u,v) \chi_{[0,s]}(v)$ acting on the function space $L^2((0,\infty))$. The celebrated Tracy-Widom formula (see \cite[Equation (1.19)]{TW94}) is
\begin{equation}\label{eq:TWbessel}
  \det\left(I-\mathcal{K}_{\textrm{Bes}}\right) = \exp \left( -\frac{1}{4} \int_0^s \log\left(\frac{s}{\tau}\right) \, \mathfrak{q}^2(\tau) d\tau \right),
\end{equation}
where the function $\mathfrak{q}(\tau)$ satisfies the following equation
\begin{equation} \label{q-pv}
  \tau(\mathfrak{q}^2-1) (\tau\mathfrak{q}')' = \mathfrak{q}(\tau\mathfrak{q}')^2 + \frac{1}{4}(\tau-\alpha^2)\mathfrak{q} + \frac{1}{4} \tau \mathfrak{q}^3 (\mathfrak{q}^2 -2)
\end{equation}
with the boundary condition
\begin{equation} \label{q-pv-bc}
  \mathfrak{q}(\tau) \sim \frac{1}{2^\alpha \Gamma(1+\alpha)} \tau^{\alpha/2}, \qquad \textrm{as } \tau \to 0^+.
\end{equation}
Note that, the equation \eqref{q-pv}  reduces to a special case of the Painlev\'e V equation via the transformation $y(\tau) = (q(\tau^2 ) -1)/(q(\tau^2 )+ 1)$. Furthermore, the following large $s$ asymptotics of $\det(I-\mathcal{K}_{\textrm{Bes}})$ is conjectured in \cite{TW94}, and later rigorously proved in \cite{DKV11} (see also \cite{Bothner16}):
\begin{equation}\label{eq:asyBesDet}
\ln \det(I-\mathcal{K}_{\textrm{Bes}}) = -\frac{1}{4}s+ \alpha s^{1/2}-\frac{\alpha^2}{4}\ln s
+  \tau_{\alpha}+ O\left(s^{-1/2}\right),\quad \textrm{as } s \to +\infty,
\end{equation}
where the constant $ \tau_{\alpha}$ is given by
\begin{equation}\label{def:tau}
\tau_\alpha=\ln \left(\frac{G(1+\alpha)}{(2\pi )^{\alpha/2}} \right)
\end{equation}
with $G(z)$ being the Barnes $G$-function.

It would then be natural to derive the large $s$ asymptotics of $\det(I-\mathcal{K}_{\textrm{PIII}})$ and to ask for its Painlev\'{e} type formula, which will be the main results of the present work stated in what follows.  We note that similar problems have been addressed for other generalizations of Bessel kernel recently in \cite{CGS,Girotti14,Stro14,Zhang17}.


\section{Statement of results}


\subsection{Large gap asymptotics}

Our first result is the following large $s$ asymptotics of the gap probability.

\begin{thm}\label{thm:asy}
For $\alpha>-1$, $s>0$ and $\lambda>0$, let $\mathcal{K}_{\textrm{PIII}}$ be the integral operator with kernel $\chi_{[0,s]}(u) K_{\textrm{PIII}}(u,v) \chi_{[0,s]}(v)$ in \eqref{eq:PIIIkernel} acting on the function space $L^2((0,\infty))$ and denote
\begin{equation}\label{def:Fnotation}
F(s;\lambda):=\ln \det(I-\mathcal{K}_{\textrm{PIII}}).
\end{equation}
Then, we have, as $s\to +\infty$,
\begin{equation}\label{eq:sasylndet}
F(s;\lambda) = -\frac{1}{4}s+ \alpha s^{1/2} -\frac{\alpha^2}{4}\ln s
+\int_0^\lambda \frac{1}{2t} \left(r(t)+\frac{\alpha^2}{2}-\frac{1}{8}\right) dt+  \tau_{\alpha}+ O\left(s^{-1/2} \right),
\end{equation}
where the function $r(t)$ is related to the Painlev\'{e} III hierarchy as stated in Proposition \ref{prop:PIII} and the constant $\tau_{\alpha}$ is given in \eqref{def:tau}.
\end{thm}

In view of the local behavior of $r(\lambda)$ near the origin given in \eqref{eq:rsmall}, the integral in \eqref{eq:sasylndet} is well-defined. Moreover, as $\lambda\to 0^+$, we recover the large $s$ asymptotics of
$\ln \det(I-\mathcal{K}_{\textrm{Bes}})$ shown in \eqref{eq:asyBesDet}. This is compatible with the fact  \eqref{eq:PIIItoBesl}, and also explains why the two constant terms in \eqref{eq:asyBesDet} and \eqref{eq:sasylndet} are the same.

\begin{remark}
In the literature, the asymptotics of Fredholm determinants of other Painlev\'{e} kernels have been
investigated in \cite{CIK10} for the kernels built in terms of the Painlev\'{e} I hierarchy, in \cite{Bot:Its2014} for the  Painlev\'{e} II kernel associated with the Hastings-McLeod solution, and quite recently in \cite{XD2017} for the
Painlev\'e XXXIV kernel. All these Painlev\'{e} kernels describe certain critical behaviors encountered in random matrix theory.
\end{remark}

\begin{remark}
We would like to point out that our asymptotic formula \eqref{eq:sasylndet} provides another concrete example that supports the so-called Forrester-Chen-Eriksen-Tracy conjecture \cite{CET95,Forrester93}. This conjecture asserts that if the mean density $\rho(x)$ of particles around a point $x^\ast$ behaves like $(x-x^\ast)^\beta$, then, for the probability $E(s)$ of emptiness of the interval $(x^\ast-s,x^\ast+s)$, we have
\begin{equation}\label{eq:conjecture}
E(s) \sim \exp\left(-Cs^{2\beta+2}\right).
\end{equation}

When $E(s)$ admits a representation in terms of a Fredholm determinant associated with certain kernel, this conjecture is supported by all the classical kernels encountered in random matrix theory, which include the sine kernel \cite{Dyson} ($\beta=0$), the Airy kernel \cite{TW94a} ($\beta=1/2$) and the Bessel kernel \eqref{eq:asyBesDet} ($\beta=-1/2$). Other evidences beyond these classical cases include higher Painlev\'{e} I kernels \cite{CIK10} ($\beta=2l+1/2$, $l\in\mathbb{N}$), the Painlev\'{e} II kernel \cite{Bot:Its2014} ($\beta=2$), the Meijer G-kernels ($\beta=-\frac{l}{l+1}$, $l\in\mathbb{N}$) and Wright's generalized Bessel kernels ($\beta=-\frac{1}{1+\theta}$, $\theta>0$) investigated in \cite{CGS}.

For the model \eqref{eq:probmeasure} considered in the present work, the limiting mean distribution of the eigenvalues as $n\to\infty$ and $t\to 0^+$ is the same as the unperturbed case, which could be described as the minimizer of an equilibrium problem with external field $V$. Due to our assumption on the regularity of the potential $V$, the density function behaves like $x^{-1/2}$ near the origin. Hence, the asymptotic formula \eqref{eq:sasylndet} indeed supports \eqref{eq:conjecture} with $\beta=-1/2$.

Finally, we highlight that the $s$-independent term in \eqref{eq:sasylndet} involves an integral with respect to certain Painlev\'{e} transcendents. This phenomena was first observed in \cite{Bot:Its2014} in the studies of Painlev\'{e} II kernel associated with the Hastings-McLeod solution, and is also confirmed in recent work \cite{XD2017} concerning the
Painlev\'e XXXIV kernel. We expect this feature is true for the large $s$ asymptotics associated with any other Painlev\'{e} kernels.
\end{remark}

\subsection{A coupled Painlev\'{e} III system}
To express a Painlev\'{e} type formula for the gap probability, we need a function $a(\lambda; s)$. This function, together with the other $k+1$ functions $b_1(\lambda;s), \ldots, b_{k+1}(\lambda;s)$, satisfies the following system of coupled ODEs involving $k+2$ equations:
\begin{equation} \label{coupled PIII}
\begin{cases}
 (\lambda -2b_1) b_{1}'' + \left(b_{1}' - \frac{1}{2}\right)^2 - (2a' - s) (\lambda -2b_1)^2 = 0, \\
  \ds\sum_{m=j-k-1}^{k+1} \left( -b_{j-m} b_{m}'' + \frac{b_{j-m}' b_{m}'}{2}  + 4 a' b_{j-m} b_{m} + 2 b_{j-m} b_{m+1} \right) = \Lambda_j, & k+2\leq j \leq 2k+2,
\end{cases}
\end{equation}
where $'=\frac{d}{d\lambda}$, $b_{k+2}(\lambda;s) := 0$ and the constants $\Lambda_j$ are given by
\begin{equation}
  \Lambda_j = \begin{cases}
    2k^2, & j = 2k+2, \\
    (-1)^k 2 \alpha k, & j= k+2, \\
    0, & k+2<j<2k+2.
  \end{cases}
\end{equation}
In addition, we also have
\begin{equation}\label{eq:a-b}
 \frac{\partial}{\partial s}a(\lambda;s)=\frac{\lambda}{2} - b_1(\lambda;s).
\end{equation}

We call the system \eqref{coupled PIII} a coupled Painlev\'{e} III system for the following reason. When $k=1$, the ODEs in \eqref{coupled PIII} reduce to
  \begin{equation} \label{coupled PIII:k=1}
    \begin{cases}
     (\lambda -2b_1) b_{1}'' + \left(b_{1}' - \frac{1}{2}\right)^2 - (2a' - s) (\lambda -2b_1)^2 = 0, \\
     -b_{2} b_{2}'' + \frac{1}{2} (b_{2}')^2 + 4 a' b_{2}^2 = 2,  \\
     b_1 b_{2}'' - b_{1}'' b_2 - \frac{b_1 (b_{2}')^2}{b_2} + b_{1}' b_2' + 2b_2^2 + 2\alpha + \frac{4b_1}{b_2}  =0.
    \end{cases}
  \end{equation}
Set $b_1 = \lambda /2$, the first equation in \eqref{coupled PIII:k=1} is satisfied automatically, while the third equation for $b_2$ now reads
  \begin{equation}
    b_2'' - \frac{ (b_2')^2}{b_2} + \frac{1}{\lambda} \left ( b_2' + 4 b_2^2 + 4\alpha \right) + \frac{4}{b_2}  = 0,
  \end{equation}
which is a Painlev\'e III equation.

For general $k\in\mathbb{N}$, if $s=0$, there exist solutions $a(\lambda)$ and $b_i(\lambda)$, $i=1,\dots,k+1$, to the system \eqref{coupled PIII} with $b_1 = \lambda /2$ (see Remark \ref{remark-b1} below). It is then straightforward to check that the ODEs in \eqref{coupled PIII} are related to the Painlev\'{e} III hierarchy \eqref{eq:PIIIequ} through the following relations
  \begin{equation}
  \begin{cases}
     a'(\lambda) = - \varrho(\lambda )/2, \\
    b_{j+1}(\lambda) = \ell_j(\lambda)/4^j, \qquad  j=0, \cdots, k.
    \end{cases}
    \end{equation}

Our second result concerns the existence of a class of special solutions to the above
coupled Painlev\'{e} III system.

\begin{thm} \label{thm:cp-solutions}
For $\alpha>-1$, $s>0$, there exist solutions  $a(\lambda;s)$ and $b_1(\lambda;s),\ldots,b_{k+1}(\lambda;s)$ to the coupled  Painlev\'{e} III system  \eqref{coupled PIII}, which are analytic for $\lambda\in(0,+\infty)$ and the function $a(\lambda;s)$ satisfies the following asymptotic behaviors as $\lambda\to 0^{+}$:
\begin{align}
a(\lambda;s)&=\frac{1-4\alpha^2}{8\lambda}+O\left(\lambda^{1+2\alpha}\right). \label{eqr: thm-asy}
\end{align}
\end{thm}

In principle, we could also derive the asymptotics of other functions $b_{i}(\lambda;s)$, $i=1,\ldots,k+1$, for the special solutions in Theorem \ref{thm:cp-solutions}. For instance, we have
\begin{align}\label{eq:asyb1}
b_1(\lambda;s)&= \frac{\lambda}{2} -
\frac{s^{\alpha} \lambda^{2\alpha+1}}{2^{2\alpha+1}  \, \Gamma(\alpha+1)^2}e^{-\frac{2}{s^k}}
\left(1+O\left(\lambda^{2(1+\alpha)}\right)\right), \quad \lambda\to 0^{+}.
\end{align}
Since it is the function $a(\lambda;s)$ that will play a role in our Painlev\'{e} type formula below, we will not discuss their asymptotics in this paper.

\subsection{A Painlev\'{e} type formula of the gap probability}
Our final result shows the existence of a Painlev\'{e} type formula of the gap probability, after proper scaling.

\begin{thm}\label{thm:TW}
With the logarithm of the gap probability $F(s;\lambda)$ defined in \eqref{def:Fnotation}, we have
\begin{equation}\label{eq:TW}
F\left(\lambda^2s;\lambda^2\right) = -\int_{0}^{\lambda}[a(\tau;s)-a(\tau;0)]d\tau,
\end{equation}
where the function $a(\lambda;s)$ is among the class of special solutions to the coupled Painlev\'{e} III system \eqref{coupled PIII} as stated in Theorem \ref{thm:cp-solutions}.

%
\end{thm}

By \eqref{eqr: thm-asy}, it is readily seen that the integral in \eqref{eq:TW} is convergent. We also emphasize that the formula \eqref{eq:TW} is consistent with the Tracy-Widom formula for Bessel kernel in the sense that it boils down to the logarithm of \eqref{eq:TWbessel} as $\lambda \to 0^+$. For convenience of the reader, we will discuss this aspect in the Appendix.

\begin{remark}
It is worthwhile to note that the relations between Fredholm determinants and other coupled Painlev\'e systems have also been established in recent studies, which generalize the classical results of Tracy and Widom \cite{TW94,TW94a}. More precisely, the coupled Painlev\'e II systems have been related to the generating function for the Airy point process in \cite{Claeys:Doe},
to the Fredholm determinants of the Painlev\'e II and the Painlev\'e XXXIV kernel in \cite{XD2017}.
A coupled Painlev\'e V system has been related to the generating function for
the Bessel point process in \cite{Char:Doera2017}.
\end{remark}

\subsection{Organization of the rest of the paper}
The rest of this paper is devoted to the proofs of our results. Due to the integrable structure of the kernel $K_{\textrm{PIII}}(u,v)$ in the sense of Its-Izergin-Korepin-Slavnov \cite{IIKS90}, the proof of Theorem \ref{thm:asy} relies on the representations of partial derivatives of Fredholm determinant in terms of the solutions of certain RH problems, established under a general framework in \cite{Bor:Dei2002,DIZ97}. In Section \ref{sec:sderivative}, we formulate a RH problem $X$ with constant jumps related to $\frac{\partial}{\partial s} F(s;\lambda)$ and then perform a Deift-Zhou steepest descent analysis \cite{DZ93} on this RH problem, which ultimately leads to the large $s$ asymptotics of $\frac{\partial}{\partial s} F(s;\lambda)$. Similarly, we study a RH problem $\widehat{X}$ associated with $\frac{\partial}{\partial \lambda} F\left(\lambda^2s;\lambda^2\right)$ and its large $s$ asymptotics in Section \ref{sec:l-derivative}. The differences between the RH problems for $X$ and $\widehat X$ lie in the boundary conditions. This asymptotics, together with the asymptotics of $\frac{\partial}{\partial s} F(s;\lambda)$, will lead to large $s$ asymptotics of $\frac{\partial}{\partial \lambda} F(s;\lambda)$. A combination of these two asymptotics of partial derivatives of $F(s;\lambda)$ gives us the proof of Theorem \ref{thm:asy}, as shown in Section \ref{sec:proofthmasy}. The derivation of the coupled Painlev\'{e} III system is given in Section \ref{sec:lax-pair}, where we formulate a model RH problem $M$ and the coupled Painlev\'{e} III system follows from the compatibility conditions of the Lax pair equations associated with this RH problem. The RH problem for $M$ is actually defined as a `scaled' version of that for $X$, and is also related to $\frac{\partial}{\partial s} F\left(\lambda^2s;\lambda^2\right)$. After performing an asymptotic analysis of this RH problem in Section \ref{sec:lambda-small}, we present the proofs of Theorems \ref{thm:cp-solutions} and \ref{thm:TW} in Section \ref{sec:proofoflastthm}. Finally, the consistency of \eqref{eq:TW} with the logarithm of \eqref{eq:TWbessel} is discussed in the Appendix.


\section{Large $s$ asymptotics of $\frac{\partial}{\partial s} F(s;\lambda)$ }
\label{sec:sderivative}

In this section, we will derive large $s$ asymptotics of $\frac{\partial}{\partial s} F(s;\lambda)$ by relating
$\frac{\partial}{\partial s} F(s;\lambda)$ to a RH problem. The strategy follows from a general framework established in \cite{Bor:Dei2002,DIZ97}. We first recall some basic facts in the context of our model.

\subsection{A Riemann-Hilbert setting}

To facilitate our computations, we shall make use of the fact that
\begin{equation}\label{eq:FtildeK}
\det(I-\mathcal{K}_{\textrm{PIII}})=\det\left(I-\widetilde{\mathcal{K}}_{\textrm{PIII}}\right),
\end{equation}
where $\widetilde{\mathcal{K}}_{\textrm{PIII}}$ is the integral operator with kernel
\begin{equation}\label{eq:tildePIII}
\chi_{[-s,0]}(u) \widetilde{K}_{\textrm{PIII}}(u,v)\chi_{[-s,0]}(v)=e^{\alpha \pi i }\chi_{[-s,0]}(u) \frac{\psi_1(v;\lambda)\psi_2(u;\lambda)  -\psi_1(u;\lambda)\psi_2(v;\lambda)}{2\pi i(u-v)}\chi_{[-s,0]}(v)
\end{equation}
acting on the function space $L^2((-\infty,0))$.

By \eqref{eq:tildePIII}, it is readily seen that
\begin{equation}
\widetilde{K}_{\textrm{PIII}}(u,v)= \frac{\vec{f}^t(u)\vec{h}(v)}{u-v},
\end{equation}
where
\begin{equation}\label{def:fh}
\vec{f}=-e^{\frac{\alpha \pi i}{2}}
\begin{pmatrix}
\psi_1 \\
\psi_2
\end{pmatrix}, \qquad
\vec{h}=\frac{e^{\frac{\alpha \pi i}{2}}}{2 \pi i}
\begin{pmatrix}
\psi_2 \\
-\psi_1
\end{pmatrix}.
\end{equation}
In addition, we have
\begin{equation}\label{eq:derivative1}
\frac{d}{ds} \ln \det\left(I-\widetilde{\mathcal{K}}_{\textrm{PIII}}\right)=-\textrm{tr}
\left(\left(I-\widetilde{\mathcal{K}}_{\textrm{PIII}}\right)^{-1}
\frac{d}{ds}\widetilde{\mathcal{K}}_{\textrm{PIII}}\right)=-R(-s,-s),
\end{equation}
where $R(u,v)$ stands for
the kernel of the resolvent operator
$$R=\left(I-\widetilde {\mathcal{K}}_{\textrm{PIII}}\right)^{-1}-I=\widetilde {\mathcal{K}}_{\textrm{PIII}}\left(I-\widetilde {\mathcal{K}}_{\textrm{PIII}}\right)^{-1}=\left(I-\widetilde {\mathcal{K}}_{\textrm{PIII}}\right)^{-1}\widetilde {\mathcal{K}}_{\textrm{PIII}}.$$
Since the operator $\widetilde {\mathcal{K}}_{\textrm{PIII}}$ is integrable, its resolvent kernel is integrable as well; cf. \cite{DIZ97,IIKS90}. Indeed, by setting
\begin{equation}\label{def:FH}
\vec{F}:=\left(I-\widetilde{\mathcal{K}}_{\textrm{PIII}}\right)^{-1}\vec{f}, \qquad \vec{H}:=\left(I-\widetilde{\mathcal{K}}_{\textrm{PIII}}\right)^{-1}\vec{h},
\end{equation}
we have
\begin{equation}\label{eq:resolventexpli}
R(u,v)=\frac{\vec{F}^t(u)\vec{H}(v)}{u-v}.
\end{equation}

It turns out that this resolvent kernel is related to the following RH problem.
\begin{rhp} \label{rhp:Y}
Let $s>0$ and $\lambda>0$. We look for a $2\times 2$ matrix-valued function
$Y(z)$ satisfying
\begin{enumerate}
\item[\rm (1)] $Y(z)$ is defined and analytic in $\mathbb{C}\setminus [-s,0]$, where the orientation is taken from the left to the right.

\item[\rm (2)] For $x\in(-s,0)$, we have
\begin{equation}\label{eq:Y-jump}
 Y_+(x)=Y_-(x)\left(I-2\pi i \vec{f}(x)\vec{h}^t(x)\right),
 \end{equation}
where $\vec{f}$ and $\vec{h}$ also depend on the parameter $\lambda$ by \eqref{def:fh}.

\item[\rm (3)] As $z \to \infty$,
\begin{equation}\label{eq:Y-infty}
 Y(z)=I+O(1/z).
 \end{equation}

\item[\rm (4)] We have
\begin{equation}\label{eq:Ylocal}
 Y(z)=\left\{
        \begin{array}{ll}
          O(1), & \hbox{as $z\to 0$,} \\
          O(\ln(z+s)), & \hbox{as $z \to -s$.}
        \end{array}
      \right.
 \end{equation}
 \end{enumerate}
\end{rhp}
By \cite{DIZ97}, it follows that
\begin{equation} \label{Y-def}
  Y(z)=I-\int_{-s}^0\frac{\vec{F}(w)\vec{h}^t(w)}{w-z}dw
\end{equation}
and
\begin{equation}\label{def:FH2}
\vec{F}(z)=Y(z)\vec{f}(z), \qquad \vec{H}(z)=\left(Y^t(z)\right)^{-1}\vec{h}(z).
\end{equation}

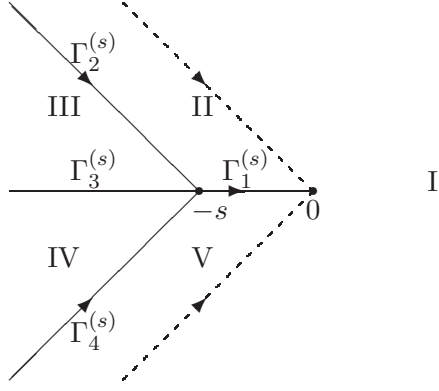
\begin{figure}[t]
\begin{center}
   \setlength{\unitlength}{1truemm}
   \begin{picture}(100,65)(-15,2)
       \dashline{0.8}(15,15)(40,40)
       \dashline{0.8}(15,65)(40,40)

       \put(40,40){\line(-1,0){40}}

       \put(25,40){\line(-1,-1){25}}
       \put(25,40){\line(-1,1){25}}
       \put(10,55){\thicklines\vector(1,-1){1}}
       \put(10,25){\thicklines\vector(1,1){1}}

        \put(25,25){\thicklines\vector(1,1){1}}

       \put(25,55){\thicklines\vector(1,-1){1}}

       \put(30,40){\thicklines\vector(1,0){1}}

       \put(39,36.3){$0$}
       \put(24,36.3){$-s$}

       \put(28,41.5){$\Gamma_1^{(s)}$}
        \put(8,57){$\Gamma_2^{(s)}$}
         \put(8,20){$\Gamma_4^{(s)}$}

       \put(8,41.5){$\Gamma_3^{(s)}$}

       \put(24,50){$\textrm{II}$}
       \put(5,50){$\textrm{III}$}
        \put(5,30){$\textrm{IV}$}

       \put(24,30){$\textrm{V}$}
       \put(55,40){$\textrm{I}$}

       \put(40,40){\thicklines\circle*{1}}
       \put(25,40){\thicklines\circle*{1}}

   \end{picture}
   \caption{Regions $\textrm{I}-\textrm{V}$ and contours $\Gamma_j^{(s)}$, $j=1,2,3,4$.}
   \label{fig:YtoX}
\end{center}
\end{figure}

Recalling the RH problem \ref{rhp:Psi} for $\Psi$, we make the following undressing transformation to arrive at a RH problem with constant jumps. Define
\begin{equation}\label{eq:YtoX}
X(z;\lambda,s)= X(z)=\left\{
       \begin{array}{ll}
         Y(z)\Psi(z), & \hbox{for $z$ in region $\textrm{I}\cup \textrm{III}\cup \textrm{IV}$,} \\
         Y(z)\Psi(z)
\begin{pmatrix}
1 & 0 \\
e^{\alpha \pi i} & 1
\end{pmatrix}, & \hbox{for $z$ in region $\textrm{II}$,} \\
         Y(z)\Psi(z)
\begin{pmatrix}
1 & 0 \\
-e^{-\alpha \pi i} & 1
\end{pmatrix}, & \hbox{for $z$ in region $\textrm{V}$,}
       \end{array}
     \right.
\end{equation}
where the regions $\textrm{I}-\textrm{V}$ are shown in Figure \ref{fig:YtoX}. Then, $X$ satisfies the following RH problem.

\begin{prop}\label{prop:RHPforX}
The function $X$ defined in \eqref{eq:YtoX} has the following properties:
\begin{enumerate}
\item[\rm (1)] $X(z)$ is defined and analytic in $\mathbb{C}\setminus \{\cup^4_{j=1}\Gamma_j^{(s)}\cup\{-s\}\}$, where
\begin{equation}\label{def:gammais}
\Gamma_1^{(s)}=(-s,0),\quad \Gamma_2^{(s)}=-s+e^{-\pi i/3}\mathbb{R}^-,\quad \Gamma_3^{(s)}=(-\infty,-s),\quad \Gamma_4^{(s)}=-s+e^{\pi i/3}\mathbb{R}^-,
\end{equation}
with $\mathbb{R}^-:=(-\infty, 0)$ and all orientations from the left to the right; see the solid lines in Figure \ref{fig:YtoX}.

\item[\rm (2)] $X$ satisfies the following jump conditions:
\begin{equation}\label{eq:X-jump}
 X_+(z)=X_-(z)
 \left\{
 \begin{array}{ll}

                               \begin{pmatrix}
e^{\alpha \pi i} & 0 \\
0 & e^{-\alpha \pi i}
\end{pmatrix}, &  z \in \Gamma_1^{(s)}, \\[.4cm]
    \left(
                               \begin{array}{cc}
                                1 & 0\\
                               e^{\alpha \pi i} & 1 \\
                                 \end{array}
                             \right),  &  z \in \Gamma_2^{(s)}, \\[.4cm]
    \left(
                               \begin{array}{cc}
                                0 & 1\\
                               -1 & 0 \\
                                 \end{array}
                             \right),  &  z \in \Gamma_3^{(s)}, \\[.4cm]
    \left(
                               \begin{array}{cc}
                                 1 & 0 \\
                                 e^{-\alpha\pi i} & 1 \\
                                 \end{array}
                             \right), &   z \in \Gamma_4^{(s)}.
 \end{array}  \right .  \end{equation}

\item[\rm (3)] As $z\to\infty$,
\begin{equation}\label{eq:X-infty}
 X(z)=
 \left( I + \frac{X_\infty}{z} + O\left( z^{-2}\right)
 \right) z^{-\frac{1}{4} \sigma_3} \frac{I + i \sigma_1}{\sqrt{2}} e^{z^{\frac{1}{2}} \sigma_3}.
   \end{equation}

\item[\rm (4)]We have
\begin{equation}\label{eq:X-origin}
 X(z)=\left\{
        \begin{array}{ll}
          O(1)(I+O(z)) e^{(-1)^{k+1}\left(\frac{\lambda}{z}\right)^k\sigma_3}z^{\frac \alpha2\sigma_3}, & \hbox{as $z\to 0$,} \\
          O(\ln(z+s)), & \hbox{as $z \to -s$.}
        \end{array}
      \right.
 \end{equation}
\end{enumerate}
\end{prop}
\begin{proof}
It suffices to check the jump condition on $\Gamma_1^{(s)}$, while the other claims follow directly from the definition \eqref{eq:YtoX}.

By \eqref{eq:psinegx} and \eqref{def:fh}, we have, for $x<0$,
\begin{equation}\label{eq:fhinPsik=1}
\vec{f}(x)=\Psi_-(x)
\begin{pmatrix}
-e^{\frac{\alpha \pi i}{2}} \\
e^{-\frac{\alpha \pi i}{2}}
\end{pmatrix}, \qquad
\vec{h}(x)=-\frac{1}{2 \pi i}
\Psi_-^{-t}(x)
\begin{pmatrix}
e^{-\frac{\alpha \pi i}{2}} \\
e^{\frac{\alpha \pi i}{2}}
\end{pmatrix}.
\end{equation}
Thus,
\begin{equation}
I-2\pi i \vec{f}(x)\vec{h}^t(x)=\Psi_-(x)
\begin{pmatrix}
0 & -e^{\alpha \pi i} \\
e^{-\alpha \pi i} & 2
\end{pmatrix}\Psi_-^{-1}(x).
\end{equation}
This, together with \eqref{eq:Y-jump} and \eqref{eq:YtoX}, implies that if $-s<x<0$,
\begin{equation}
\begin{aligned}
X_+(x)& = Y_+(x)\Psi_+(x)
\begin{pmatrix}
1 & 0 \\
e^{\alpha \pi i} & 1
\end{pmatrix} \\
&= Y_-(x)\Psi_-(x)
\begin{pmatrix}
0 & -e^{\alpha \pi i} \\
e^{-\alpha \pi i} & 2
\end{pmatrix}\Psi_-^{-1}(x)\Psi_+(x)
\begin{pmatrix}
1 & 0 \\
e^{\alpha \pi i} & 1
\end{pmatrix}
\\
&=X_-(x)\begin{pmatrix}
1 & 0 \\
e^{-\alpha \pi i} & 1
\end{pmatrix}\begin{pmatrix}
0 & -e^{\alpha \pi i} \\
e^{-\alpha \pi i} & 2
\end{pmatrix}\begin{pmatrix}
0 & 1 \\
-1 & 0
\end{pmatrix}
\begin{pmatrix}
1 & 0 \\
e^{\alpha \pi i} & 1
\end{pmatrix}
\\
&=X_-(x)\begin{pmatrix}
e^{\alpha \pi i} & 0 \\
0 & e^{-\alpha \pi i}
\end{pmatrix},
\end{aligned}
\end{equation}
as desired.
\end{proof}

\begin{remark}
Solvability of the RH problem for $X$ follows from \eqref{eq:MandX} and Theorem \ref{Solvability} below.
\end{remark}


The connection between $X$ and $\frac{\partial}{\partial s} F(s;\lambda)$ is stated in the following proposition.
\begin{prop}\label{prop:Xandgap}
We have
\begin{equation}\label{eq:derivativeinX}
\frac{\partial}{\partial s} F(s;\lambda)=\frac{d}{ds} \ln \det\left(I-\widetilde{\mathcal{K}}_{\textrm{PIII}}\right)=-
\frac{e^{\alpha \pi i}}{2 \pi i} \lim_{z \to -s} \left(X^{-1}(z)X'(z)\right)_{21},
\end{equation}
where the limit is taken as $z$ tends to $-s$ in region $\textrm{V}$.
\end{prop}
\begin{proof}
For $z$ belonging to region V, we see from \eqref{def:FH2}, \eqref{eq:YtoX} and \eqref{eq:fhinPsik=1} that
\begin{equation}\label{eq:FHexplicit}
\vec{F}(z)=X(z)
\begin{pmatrix}
-e^{\alpha \pi i/2} \\
0
\end{pmatrix},\qquad \vec{H}(z)=-\frac{X^{-t}(z)}{2 \pi i}
\begin{pmatrix}
0 \\
e^{\alpha \pi i/2}
\end{pmatrix}.
\end{equation}
A combination of L'H\^{o}pital's rule, \eqref{eq:FtildeK}, \eqref{eq:derivative1}, \eqref{eq:resolventexpli} and \eqref{eq:FHexplicit} then gives us \eqref{eq:derivativeinX}.
\end{proof}

\subsection{Asymptotic analysis of the Riemann-Hilbert problem for $X$}
\label{sec:sasyanalysis}
In this section, we shall perform the Deift-Zhou steepest descent analysis to the RH problem for $X$ as $s \to +\infty$. It consists of a series of explicit and invertible transformations which leads to a RH problem tending to the identity matrix as $s \to +\infty$.

\subsubsection{$X\to T$: Rescaling}
Define
\begin{equation}\label{eq:XtoT}
T(z)=X(sz).
\end{equation}
It is then straightforward to check that $T$ satisfies a RH problem as follows:
\begin{rhp}
The function $T$ defined in \eqref{eq:XtoT} has the following properties:
\begin{enumerate}
\item[\rm (1)] $T(z)$ is defined and analytic in $\mathbb{C} \setminus
\{\cup^4_{j=1}\Gamma_j^{(1)}\cup\{-1\}\}$, where the contours $\Gamma_{j}^{(1)}$ are defined in \eqref{def:gammais} with $s=1$, as shown in
Figure \ref{fig:contour-T}.

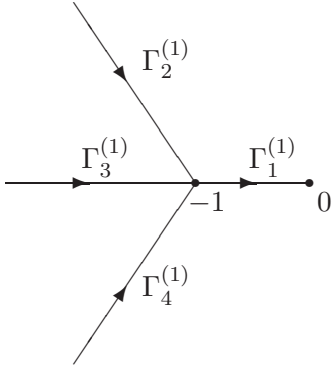
\begin{figure}[h]
\begin{center}
   \setlength{\unitlength}{1truemm}
   \begin{picture}(100,65)(-25,2)
       \put(40,40){\line(-1,0){40}}

       \put(25,40){\line(-2,-3){16}}
       \put(25,40){\line(-2,3){16}}

       \put(15,55){\thicklines\vector(2,-3){1}}
       \put(15,25){\thicklines\vector(2,3){1}}

       \put(32,40){\thicklines\vector(1,0){1}}
       \put(10,40){\thicklines\vector(1,0){1}}

       \put(41,36.3){$0$}
       \put(24,36.3){$-1$}
       \put(32,42){$\Gamma_1^{(1)}$}
       \put(18,55){$\Gamma_2^{(1)}$}
       \put(10,42){$\Gamma_3^{(1)}$}
       \put(18,25){$\Gamma_4^{(1)}$}


       \put(40,40){\thicklines\circle*{1}}
       \put(25,40){\thicklines\circle*{1}}

   \end{picture}
   \vspace{-16mm}
   \caption{Contours $\Gamma_j^{(1)}$, $j=1,2,3,4$, for the RH problem for $T$.}
   \label{fig:contour-T}
\end{center}
\end{figure}

\item[\rm (2)] $T$ satisfies the following jump conditions:
\begin{equation}\label{eq:T-jump}
 T_+(z)=T_-(z)
 \left\{
 \begin{array}{ll}
\begin{pmatrix}
e^{\alpha \pi i} & 0 \\
0 & e^{-\alpha \pi i}
\end{pmatrix}, &  z \in \Gamma_1^{(1)}, \\[.4cm]
    \left(
                               \begin{array}{cc}
                                1 & 0\\
                               e^{\alpha \pi i} & 1 \\
                                 \end{array}
                             \right),  &  z \in \Gamma_2^{(1)}, \\[.4cm]
    \left(
                               \begin{array}{cc}
                                0 & 1\\
                               -1 & 0 \\
                                 \end{array}
                             \right),  &  z \in \Gamma_3^{(1)}, \\[.4cm]
    \left(
                               \begin{array}{cc}
                                 1 & 0 \\
                                 e^{-\alpha\pi i} & 1 \\
                                 \end{array}
                             \right), &   z \in \Gamma_4^{(1)}.
 \end{array}  \right .  \end{equation}

\item[\rm (3)] As $z\to\infty$,
\begin{equation}
 T(z)=
 \left( I + O\left(\frac{1}{z}\right)
 \right) (sz)^{-\frac{1}{4} \sigma_3} \frac{I + i \sigma_1}{\sqrt{2}} e^{(sz)^{\frac{1}{2}} \sigma_3}.
   \end{equation}

\item[\rm (4)]For $z \notin (-\infty, 0]$, we have
\begin{equation}\label{eq:T-origin}
 T(z)=\left\{
        \begin{array}{ll}
          O(1)(I+O(z)) e^{(-1)^{k+1}\left(\frac{\lambda}{sz}\right)^k\sigma_3}z^{\frac \alpha2\sigma_3}, & \hbox{as $z\to 0$,} \\
          O(\ln(z+1)), & \hbox{as $z \to -1$.}
        \end{array}
      \right.
 \end{equation}
\end{enumerate}
\end{rhp}

\subsubsection{$T\to S$: Normalization at $\infty$ and $0$}
Define
\begin{equation}\label{def:g}
g(z)=\sqrt{z+1}\left(1+ (-1)^{k+1}\frac{\lambda^k}{s^{k+1/2}} \sum_{j=1}^k \frac{c_j}{z^j} \right), \qquad z\in\mathbb{C}\setminus(-\infty, -1],
\end{equation}
where
\begin{equation} \label{g-exp-cj}
  c_k=1 \quad \textrm{and} \quad \sum_{j=0}^m \binom{1/2}{m-j} c_{k-j} = 0 \qquad \textrm{for } m =1, \ldots, k-1.
\end{equation}
From the above definition, it is readily seen that
$$g_{+}(x)+g_-(x)=0 \qquad \textrm{for } x<-1,$$
and
\begin{equation}
  g(z)= \left\{
             \begin{array}{ll}
               \sqrt{z} + \mathfrak{g}_1 z^{-1/2} + O(z^{-3/2}), & \hbox{as $z\to \infty$,} \\
               (-1)^{k+1}\frac{\lambda^k}{s^{k+1/2}z^k} +  O(1), & \hbox{as $z \to 0$,}
             \end{array}
           \right.
\end{equation}
with
\begin{equation} \label{contant-g1}
  \mathfrak{g}_1= \frac{1}{2} + (-1)^{k+1}\frac{\lambda^k}{s^{k+1/2}} c_1.
\end{equation}

The second transformation is then defined by
\begin{equation}\label{eq:TtoS}
S(z)=  s^{\frac{1}{4}\sigma_3}  \begin{pmatrix}
   1 & 0 \\ i s \, \mathfrak{g}_1 & 1
 \end{pmatrix}  T(z)e^{-\sqrt{s}g(z)\sigma_3},
\end{equation}
where $\mathfrak{g}_1$ is given in \eqref{contant-g1}.
It is straightforward to check that $S$ satisfies the following RH problem.
\begin{rhp}
The function $S$ defined in \eqref{eq:TtoS} has the following properties:
\begin{enumerate}
\item[\rm (1)] $S(z)$ is defined and analytic in $\mathbb{C} \setminus
\{\cup^4_{j=1}\Gamma_j^{(1)}\cup\{-1\}\}$; see Figure \ref{fig:contour-T}.

\item[\rm (2)] $S$ satisfies the following jump conditions:
\begin{equation}\label{eq:S-jump}
 S_+(z)=S_-(z)
 \left\{
 \begin{array}{ll}
\begin{pmatrix}
e^{\alpha \pi i} & 0 \\
0 & e^{-\alpha \pi i}
\end{pmatrix}, &  z \in \Gamma_1^{(1)}, \\[.4cm]
    \left(
                               \begin{array}{cc}
                                1 & 0\\
                               e^{\alpha \pi i-2\sqrt{s}g(z)} & 1 \\
                                 \end{array}
                             \right),  &  z \in \Gamma_2^{(1)}, \\[.4cm]
    \left(
                               \begin{array}{cc}
                                0 & 1\\
                               -1 & 0 \\
                                 \end{array}
                             \right),  &  z \in \Gamma_3^{(1)}, \\[.4cm]
    \left(
                               \begin{array}{cc}
                                 1 & 0 \\
                                 e^{-\alpha\pi i-2\sqrt{s}g(z)} & 1 \\
                                 \end{array}
                             \right), &   z \in \Gamma_4^{(1)}.
 \end{array}  \right .  \end{equation}

\item[\rm (3)] As $z\to\infty$,
\begin{equation}\label{eq:S-infty}
 S(z)=
 \left( I + O\left(\frac{1}{z}\right)
 \right)z^{-\frac{1}{4} \sigma_3} \frac{I + i \sigma_1}{\sqrt{2}}.
   \end{equation}

\item[\rm (4)]We have
\begin{equation}\label{eq:S-origin}
 S(z)=\left\{
        \begin{array}{ll}
          O(1)z^{\frac \alpha2\sigma_3}, & \hbox{as $z\to 0$,} \\
          O(\ln(z+1)), & \hbox{as $z \to -1$.}
        \end{array}
      \right.
 \end{equation}
\end{enumerate}
\end{rhp}

By \eqref{def:g}, one verifies that for sufficiently large positive $s$,
$$\Re g(z)>0, \qquad z\in\Gamma_2^{(1)}\cup\Gamma_4^{(1)}.$$
Hence, the jump matrix of $S$ on $\Gamma_2^{(1)}\cup\Gamma_4^{(1)}$ tends to the identity matrix exponentially fast as $s\to+ \infty$. This uniform convergence is not valid anymore for $z$ near $-1$.

\subsubsection{Outer parametrix}
By ignoring the exponentially small terms in the jumps \eqref{eq:S-jump} and a neighborhood of $-1$, we are led to the following RH problem for the outer parametrix $N$.

\begin{rhp} \label{rhp:N}
We look for a $2\times 2$ matrix-valued function
$N(z)$ satisfying
\begin{enumerate}
\item[\rm (1)] $N(z)$ is defined and analytic in $\mathbb{C}\setminus (-\infty,0]$.
\item[\rm (2)] For $x<0$, we have
\begin{equation}\label{eq:N-jump}
N_+(x)=N_-(x)
 \left\{
 \begin{array}{ll}
    \left(
                               \begin{array}{cc}
                                 e^{\alpha \pi i} & 0 \\
                               0 &  e^{-\alpha \pi i} \\
                                 \end{array}
                             \right), &  x \in (-1,0), \\[.4cm]
    \left(
                               \begin{array}{cc}
                                0 & 1\\
                               -1 & 0 \\
                                 \end{array}
                             \right),  &  x \in (-\infty,-1).
 \end{array}  \right .  \end{equation}

\item[\rm (3)] As $z\to\infty$,
\begin{equation}
N(z)=
\left( I + O\left(\frac{1}{z}\right)\right)z^{-\frac{1}{4} \sigma_3} \frac{I + i \sigma_1}{\sqrt{2}}.
\end{equation}

\item[\rm (4)]As $z\to 0$,
\begin{equation}\label{eq:N-origin}
N(z)=O(1)z^{\frac \alpha2\sigma_3}.
  \end{equation}
\end{enumerate}
\end{rhp}
The above RH problem can be solved explicitly and the solution is given by (cf. \cite{IKO2009})
\begin{equation}\label{Nsolution}
N(z)=\left(
                                                   \begin{array}{cc}
                                                     1 & 0 \\
                                                      i \alpha & 1 \\
                                                   \end{array}
                                                 \right)
(z+1)^{-\frac 14\sigma_3}\frac 1{\sqrt{2}}(I+i\sigma_1)\left(\frac{\sqrt{z+1}+1}{\sqrt{z+1}-1}\right)^{-\frac \alpha 2 \sigma_3}, \ z \in \mathbb{C} \setminus (-\infty, 0],
\end{equation}
where the branches are chosen as $\arg(z+1)\in (-\pi, \pi)$,  $\arg z \in (-\pi, \pi)$, and such that the last factor in \eqref{Nsolution} behaves
\begin{equation}
  \left(\frac{\sqrt{z+1}+1}{\sqrt{z+1}-1}\right)^{-\frac \alpha 2 \sigma_3}=\left (1+\frac {\alpha^2}{2z}\right ) I -\frac \alpha{\sqrt {z}} \sigma_3 +O\left (z^{-3/2}\right )~~\textrm{as $z \to \infty$}.
\end{equation}

\subsubsection{Local parametrix near $-1$}\label{sec:localpara}
Let $D_{-1}$ be a small, open neighborhood near $-1$ with fixed radius. Due to the fact that the convergence of the jump matrices to the identity matrices on  $\Gamma_2^{(1)} \cup \Gamma_4^{(1)}$ is not uniform near $-1$, we intend to find a function $P_{-1}$ satisfying a RH problem as follows.

\begin{rhp} \label{rhp:P1}
We look for a $2\times 2$ matrix-valued function
$P_{-1}(z)$ satisfying
\begin{enumerate}
\item[\rm (1)] $P_{-1}(z)$ is defined and analytic in  $\overline{D_{-1}} \setminus
\{\cup^4_{j=1}\Gamma_j^{(1)}\cup\{-1\}\}$.
\item[\rm (2)] $P_{-1}(z)$ satisfies the same jump conditions \eqref{eq:S-jump} as $S$ for $z \in D_{-1}\cap \{\cup^4_{j=1}\Gamma_j^{(1)}\}.$

\item[\rm (3)] As $s \to +\infty$, $P_{-1}$ matches $N(z)$ on the boundary $\partial D_{-1}$ of $D_{-1}$, i.e.,
\begin{equation}\label{eq:mathcingp1}
P_{-1}(z)=
N(z)(I+o(1)), \qquad z\in \partial D_{-1}.
\end{equation}

\end{enumerate}
\end{rhp}

To solve the above RH problem, we note that, as $z\to -1$,
\begin{equation}\label{eq:glocal1}
g(z)=\sqrt{z+1}   \left[1-  \frac{\lambda^k}{s^{k+1/2}} \left( \eta_0 + \eta_1 (z+1) + O(z+1)^2  \right) \right] ,
\end{equation}
with
\begin{equation} \label{cont-eta0}
  \eta_0 = \sum_{j=1}^k (-1)^{k+j} c_j, \qquad \eta_1 = \sum_{j=1}^k (-1)^{k+j} j c_j,
\end{equation}
where $c_j$ are determined in \eqref{g-exp-cj}. This local behavior suggests us to construct $P_{-1}$ with the help of the Bessel parametrix $\Phi^{(\alpha)}_{\textrm{Bes}}$, which solves the following model RH problem.
\begin{rhp}\label{rhp:Bessel}
The $2\times 2$ matrix-valued function $\Phi^{(\alpha)}_{\textrm{Bes}}(\zeta)$ has the following properties:
\begin{enumerate}
\item[\rm (1)] $\Phi^{(\alpha)}_{\textrm{Bes}}(\zeta)$ is defined and analytic in $\mathbb{C}\setminus \{\cup^3_{j=1}\Sigma_j\cup\{0\}\}$, where the contours $\Sigma_j$, $j=1,2,3$,  are illustrated in Figure \ref{fig:jumpsPsi}.

\item[\rm (2)] $\Phi^{(\alpha)}_{\textrm{Bes}}(\zeta)$ satisfies the following jump conditions:
\begin{equation}\label{eq:Bes-jump}
 \Phi^{(\alpha)}_{\textrm{Bes},+}(\zeta)=\Phi^{(\alpha)}_{\textrm{Bes},-}(\zeta)
 \left\{
 \begin{array}{ll}
    \left(
                               \begin{array}{cc}
                                1 & 0\\
                               e^{\alpha \pi i} & 1 \\
                                 \end{array}
                             \right),  &  \zeta \in \Sigma_1, \\[.4cm]
    \left(
                               \begin{array}{cc}
                                0 & 1\\
                               -1 & 0 \\
                                 \end{array}
                             \right),  &  \zeta \in \Sigma_2, \\[.4cm]
    \left(
                               \begin{array}{cc}
                                 1 & 0 \\
                                 e^{-\alpha\pi i} & 1 \\
                                 \end{array}
                             \right), &   \zeta \in \Sigma_3.
 \end{array}  \right .  \end{equation}

\item[\rm (3)] As $\zeta \to \infty$,
\begin{equation}\label{eq:Besl-infty}
 \Phi^{(\alpha)}_{\textrm{Bes}}(\zeta)=
 (4 \pi^2 \zeta )^{-\frac{1}{4} \sigma_3} \frac{I + i \sigma_1}{\sqrt{2}} \left( I + \frac{1}{16\sqrt{\zeta}}  \begin{pmatrix}
   -1-4\alpha^2 & -2i \\
   -2i & 1+4\alpha^2
 \end{pmatrix} +  O\left(\frac{1}{\zeta}\right)
 \right)e^{2\sqrt{\zeta}\sigma_3}.
   \end{equation}

\item[\rm (4)]The matrix function $\Phi^{(\alpha)}_{\textrm{Bes}}(\zeta)$ has the following local behavior near the origin.
    If $\alpha<0$,
    \begin{equation}
\Phi^{(\alpha)}_{\textrm{Bes}}(\zeta)=
O \begin{pmatrix}
|\zeta|^{\alpha/2} & |\zeta|^{\alpha/2}
\\
|\zeta|^{\alpha/2} & |\zeta|^{\alpha/2}
\end{pmatrix}, \qquad \textrm{as $\zeta \to 0$}.
\end{equation}
If $\alpha=0$,
    \begin{equation}
\Phi^{(\alpha)}_{\textrm{Bes}}(\zeta)=
O \begin{pmatrix}
\ln|\zeta| & \ln|\zeta|
\\
\ln|\zeta| & \ln|\zeta|
\end{pmatrix}, \qquad \textrm{as $\zeta \to 0$}.
\end{equation}
If $\alpha>0$,
     \begin{equation}
\Phi^{(\alpha)}_{\textrm{Bes}}(\zeta)= \begin{cases}
  O\begin{pmatrix}
|\zeta|^{\alpha/2} & |\zeta|^{-\alpha/2}
\\
|\zeta|^{\alpha/2} & |\zeta|^{-\alpha/2}
\end{pmatrix}, & \textrm{as $\zeta \to 0$ and $|\arg \zeta|<2 \pi/3$, } \\[.4cm]
  O\begin{pmatrix}
|\zeta|^{-\alpha/2} & |\zeta|^{-\alpha/2}
\\
|\zeta|^{-\alpha/2} & |\zeta|^{-\alpha/2}
\end{pmatrix}, & \textrm{as $\zeta \to 0$ and $2 \pi /3 <|\arg \zeta|< \pi $. }
  \end{cases}
\end{equation}
\end{enumerate}
\end{rhp}

By \cite{KMVV2004}, the solution to RH problem \ref{rhp:Bessel} is given by
\begin{equation}\label{eq:phiBes}
\Phi^{(\alpha)}_{\textrm{Bes}}(\zeta)=\left\{
                             \begin{array}{ll}
                               \begin{pmatrix}
I_\al(2\zeta^{1/2}) & \frac{i}{\pi}K_\al(2\zeta^{1/2}) \\
2\pi i\zeta^{1/2}I'_\al(2\zeta^{1/2}) &
-2\zeta^{1/2}K_\al'(2\zeta^{1/2})
\end{pmatrix}, & \hbox{for $|\arg \zeta|<2\pi/3$,} \\
                              \begin{pmatrix}
I_\al(2\zeta^{1/2}) & \frac{i}{\pi}K_\al(2\zeta^{1/2}) \\
2\pi i\zeta^{1/2}I'_\al(2\zeta^{1/2}) &
-2\zeta^{1/2}K_\al'(2\zeta^{1/2})
\end{pmatrix}\left(
                               \begin{array}{cc}
                                1 & 0\\
                               -e^{\alpha \pi i} & 1 \\
                                 \end{array}
                             \right), & \hbox{for $2\pi/3<\arg \zeta<\pi$,} \\
                                \begin{pmatrix}
I_\al(2\zeta^{1/2}) & \frac{i}{\pi}K_\al(2\zeta^{1/2}) \\
2\pi i\zeta^{1/2}I'_\al(2\zeta^{1/2}) &
-2\zeta^{1/2}K_\al'(2\zeta^{1/2})
\end{pmatrix}\left(
                               \begin{array}{cc}
                                 1 & 0 \\
                                 e^{-\alpha\pi i} & 1 \\
                                 \end{array}
                             \right), & \hbox{for $-\pi<\arg \zeta<-2\pi/3$.}
                             \end{array}
                           \right.
\end{equation}
where $I_\alpha$ and $K_\alpha$ denote the usual modified Bessel functions \cite{AS} with a cut along the negative real axis.

Recalling the local behaviour of $S(z)$ near $-1$ in \eqref{eq:S-origin}, the local parametrix $P_{-1}(z)$ is constructed in terms of the Bessel parametrix of order 0 as follows:
\begin{equation} \label{def:P-1}
P_{-1}(z)=E_{-1}(z)\Phi^{(0)}_{\textrm{Bes}}\left(\frac{1}{4}s g^2(z)\right)\left\{
                 \begin{array}{ll}
                   e^{(-\sqrt{s}g(z)+\frac{\alpha \pi i}{2})\sigma_3}, & \hbox{for $D_{-1}\cap \{z~|~\Im z > 0 \}$,} \\
                   e^{(-\sqrt{s}g(z)-\frac{\alpha \pi i}{2})\sigma_3}, & \hbox{for $D_{-1}\cap \{z~|~\Im z < 0 \}$,}
                 \end{array}
               \right.
\end{equation}
where
\begin{equation}\label{def:E}
E_{-1}(z):=N(z)\left\{\begin{array}{ll}
                   e^{-\frac{\alpha \pi i}{2}\sigma_3}\frac{I - i \sigma_1}{\sqrt{2}}\left(\pi^2 s g^2(z)\right)^{\frac{1}{4}\sigma_3}, & \hbox{for $\Im z > 0 $,} \\
                   e^{\frac{\alpha \pi i}{2}\sigma_3}\frac{I - i \sigma_1}{\sqrt{2}}\left(\pi^2 s g^2(z)\right)^{\frac{1}{4}\sigma_3}, & \hbox{for $ \Im z < 0 $,}
                 \end{array}
 \right.
\end{equation}
and $N(z)$ is defined in \eqref{Nsolution}.

To show $P_{-1}(z)$ defined in \eqref{def:P-1} indeed satisfies RH problem \ref{rhp:P1}, we first note from \eqref{eq:glocal1} that $g^2(z)$ is a conformal mapping of $D_{-1}$ that maps $-1$ to $0$, preserves the real line, and maps $\Gamma_2^{(1)}\cup\Gamma_4^{(1)}$ onto $\Sigma_1\cup\Sigma_3$. It is then straightforward to check from \eqref{def:P-1} and \eqref{eq:Bes-jump} that $P_{-1}$ satisfies the jump conditions stated in item $(2)$ of RH problem \ref{rhp:P1}, provided we can show the prefactor $E_{-1}(z)$ is analytic in $D_{-1}$. To see this, it suffices to check that the prefactor has no jump on $D_{-1}\cap \{\Gamma_1^{(1)}\cup\Gamma_3^{(1)}\}$. Note that $g^{1/2}(z)$ is analytic in $\mathbb{C}\setminus(-\infty,-1]$, if $x\in D_{-1}\cap \Gamma_1^{(1)}$, we see from \eqref{def:E} and \eqref{eq:N-jump} that
\begin{align}
E_{-1,+}(x)&=N_+(x)e^{-\frac{\alpha \pi i}{2}\sigma_3}\frac{I - i \sigma_1}{\sqrt{2}}\left(\pi^2 s g^2(x)\right)^{\frac{1}{4}\sigma_3}
\nonumber \\
&=N_-(x)e^{+\frac{\alpha \pi i}{2}\sigma_3}\frac{I - i \sigma_1}{\sqrt{2}}\left(\pi^2 s g^2(x)\right)^{\frac{1}{4}\sigma_3}=E_{-1,-}(x).
\end{align}
If $x\in D_{-1}\cap \Gamma_3^{(1)}$, again by \eqref{def:E} and \eqref{eq:N-jump}, it follows that
\begin{align}
&E_{-1,+}(x)E^{-1}_{-1,-}(x)
\\
&=N_+(x)e^{-\frac{\alpha \pi i}{2}\sigma_3}\frac{I - i \sigma_1}{\sqrt{2}}\left(\pi^2 s g^2(x)\right)^{\frac{1}{4}\sigma_3}\left(\pi^2 s g^2(x)\right)^{-\frac{1}{4}\sigma_3}\frac{I + i \sigma_1}{\sqrt{2}}
e^{-\frac{\alpha \pi i}{2}\sigma_3}N_-^{-1}(x)
\nonumber \\
&=N_+(x)e^{-\frac{\alpha \pi i}{2}\sigma_3}\frac{I - i \sigma_1}{\sqrt{2}}
\begin{pmatrix}
i & 0
\\
0 & -i
\end{pmatrix}\frac{I + i \sigma_1}{\sqrt{2}}
e^{-\frac{\alpha \pi i}{2}\sigma_3}
\begin{pmatrix}
0 & 1 \\
-1 & 0
\end{pmatrix}
N_+^{-1}(x)
\nonumber
\\
&=N_+(x)e^{-\frac{\alpha \pi i}{2}\sigma_3}
\begin{pmatrix}
0 & -1
\\
1 & 0
\end{pmatrix}
e^{-\frac{\alpha \pi i}{2}\sigma_3}
\begin{pmatrix}
0 & 1 \\
-1 & 0
\end{pmatrix}
N_+^{-1}(x)=I.
\end{align}

Finally, for the matching condition \eqref{eq:mathcingp1}, we observe that
\begin{equation}
  \left(\frac{\sqrt{z+1}+1}{\sqrt{z+1}-1}\right)^{-\frac \alpha 2} = \rho(z) \begin{cases}
    e^{\frac{\alpha \pi i}{2} }, & \Im z >0 \\ e^{-\frac{\alpha \pi i}{2} }, & \Im z < 0,
  \end{cases}
\end{equation}
where
\begin{equation}\label{def:rho}
\rho(z):=\left(\frac {1-\sqrt{z+1}}{1+\sqrt{z+1}}\right)^{\alpha/2}, \quad z\in\mathbb{C}\setminus((-\infty,-1],[0,+\infty)).
\end{equation}
This, together with \eqref{Nsolution}, \eqref{eq:Besl-infty} and \eqref{def:P-1} and a straightforward calculation, implies that, for $z\in\partial D_{-1}$ and large positive $s$,
\begin{align}
&P_{-1}(z) N(z)^{-1}
\nonumber
\\
&=N(z) e^{\mp\frac{\alpha \pi i}{2} \sigma_3} \frac{I-i\sigma_1}{\sqrt{2}} \left(\pi^2 s g^2(z)\right)^{\frac{1}{4} \sigma_3} \Phi^{(0)}_{\textrm{Bes}}\left(\frac{1}{4}s g^2(z)\right) e^{-\sqrt{s}g(z)\sigma_3}  e^{\pm \frac{\alpha \pi i}{2} \sigma_3} N(z)^{-1}  \nonumber \\
  &=  I+\frac{1}{s^{1/2}} \begin{pmatrix}
    1 & 0 \\ i \alpha & 1
  \end{pmatrix}  J_1(z)
  \begin{pmatrix}
    1 & 0 \\ -i \alpha & 1
  \end{pmatrix} + O(s^{-1}), \label{R-jump-estimate}
\end{align}
where
\begin{align} \label{Jump-J1}
  J_{1}(z)=\left(
  \begin{array}{cc}
    \frac {\rho(z)^{-2}-\rho(z)^2}{8 g(z)} & \frac{-i((\rho(z)+\rho(z)^{-1})^2-3)}{8 \sqrt{z+1} \, g(z)}  \\
    \frac{-i((\rho(z)+\rho(z)^{-1})^2-1) \sqrt{z+1}}{8 g(z) } &\frac {\rho(z)^2-\rho(z)^{-2}}{8 g(z)} \\
  \end{array}
\right),
\end{align}
as required.

We conclude this section by evaluating $E_{-1}(-1)$ for later use. By \eqref{Nsolution} and \eqref{def:E}, it is readily seen that
\begin{equation} \label{relation:e-etilde}
  E_{-1}(z) = \begin{pmatrix}
    1 & 0 \\ i \alpha & 1
  \end{pmatrix} \widetilde{E}_{-1}(z),
\end{equation}
where
\begin{align}\label{eq:tildeE2}
  \widetilde{E}_{-1}(z) & = (z+1)^{-\frac 14\sigma_3} \frac{I+i\sigma_1}{\sqrt{2}} \rho(z)^{\sigma_3} \frac{I-i\sigma_1}{\sqrt{2}} \left(\pi^2 s g^2(z)\right)^{\frac{1}{4}\sigma_3} \nonumber \\
  & = (z+1)^{-\frac 14\sigma_3}  \left( \frac{\rho(z) + \rho^{-1}(z)}{2} I + \frac{\rho(z) - \rho^{-1}(z)}{2} \sigma_2 \right) \left(\pi^2 s g^2(z)\right)^{\frac{1}{4}\sigma_3}
\end{align}
with $\sigma_2 =
\begin{pmatrix}
  0 & -i \\ i & 0
\end{pmatrix}$.
Since
\begin{equation}\label{h-estimate}
\rho(z)=1-\alpha (z+1)^{\frac 12}+ \frac{\alpha^2}{2} (z+1) +O\left(\left(z+1\right)^{3/2}\right), \quad z\to-1,
\end{equation}
it then follows from \eqref{eq:tildeE2} and the above formula that
\begin{equation} \label{etilde1}
\widetilde{E}_{-1}(-1) = \begin{pmatrix}
    \sqrt{\pi} \left(\frac{s^{k+1/2} - \lambda^k \eta_0}{s^k}\right)^{1/2} & \frac{i \alpha}{\sqrt{\pi}} \left(\frac{s^k}{s^{k+1/2} - \lambda^k \eta_0}\right)^{1/2} \\
    0 & \frac{1}{\sqrt{\pi}} \left(\frac{s^k}{s^{k+1/2} - \lambda^k \eta_0}\right)^{1/2}
  \end{pmatrix},
\end{equation}
and
\begin{equation} \label{etilde1prime}
 \widetilde{E}_{-1}'(-1) = \begin{pmatrix}
    \ast & \ast \\ -i \alpha \sqrt{\pi} \left(\frac{s^{k+1/2} - \lambda^k \eta_0}{s^k}\right)^{1/2} & \ast
  \end{pmatrix},
\end{equation}
where $\eta_0$ is given in \eqref{cont-eta0} and $\ast$ denotes certain unimportant entries.

\subsubsection{Final transformation}
Our final transformation is defined by
\begin{equation}\label{def:R}
R(z)=\left\{
       \begin{array}{ll}
         S(z)N(z)^{-1}, & \hbox{for $z \in \mathbb{C}\setminus D_{-1} $,} \\
         S(z)P_{-1}(z)^{-1}, & \hbox{for $z \in D_{-1}$.}
       \end{array}
     \right.
\end{equation}
It is then easily seen that $R$ satisfies the following RH problem.

\begin{rhp}\label{rhp:R}
The $2\times 2$ matrix-valued function $R(z)$ defined in \eqref{def:R} has the following properties:
\begin{enumerate}
\item[\rm (1)]  $R(z)$ is analytic in $\mathbb{C} \setminus \Sigma_{R}$,
where the contour $\Sigma_R$ is shown in Figure \ref{fig:ContourR}.

\item[\rm (2)]  $R(z)$ satisfies the jump condition  $$ R_+(z)=R_-(z)J_R(z), \qquad z\in \Sigma_R, $$ where
  \begin{equation}
                     J_{R}(z)=\left\{
                                      \begin{array}{ll}
                                        P_{-1}(z) N(z)^{-1}, & \hbox{for $z \in \partial D_{-1}$,} \\
                                        N(z) J_S(z) N(z)^{-1}, & \hbox{for $ z \in \Sigma_{R} \setminus \partial D_{-1}$.}
                                      \end{array}
\right.
  \end{equation}
\item[\rm (3)] As $z \to \infty$,
$$R(z)=I+O(1/z).$$
\end{enumerate}
\end{rhp}
%
%

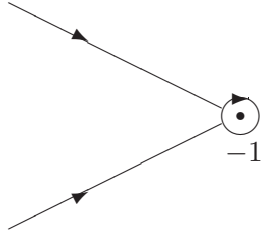
\begin{figure}[t]
\begin{center}
   \setlength{\unitlength}{1truemm}
   \begin{picture}(100,70)(-5,2)
       \put(37.5,41){\line(-2,1){28}}
       \put(37.5,39){\line(-2,-1){28}}
       \put(40,40){\thicklines\circle*{1}}
       \put(40,40){\circle{5}}

       \put(41,42.3){\thicklines\vector(1,0){.0001}}
       \put(38,34){$-1$}


       \put(20,50){\thicklines\vector(2,-1){.0001}}
       \put(20,30){\thicklines\vector(2,1){.0001}}


  \end{picture}
   \vspace{-16mm}
   \caption{Contour  $\Sigma_R$ for the RH problem \ref{rhp:R} for $R$.}
   \label{fig:ContourR}
\end{center}
\end{figure}
For $z\in \Sigma_{R} \setminus \overline{D_{-1}} $, there exits some constant $c>0$ such that
\begin{equation}
J_R(z)=I+O\left(e^{-c\sqrt{s}}\right), \qquad \textrm{as $s\to+\infty$,}
\end{equation}
uniformly valid for $\lambda$ in any compact subset of $(0,+\infty)$. This, together with \eqref{R-jump-estimate} and standard analysis (cf. \cite{Deift99book,DZ93}), implies that $R(z)$ admits an asymptotic expansion of the following form
\begin{equation}\label{R-expand}
R(z)=I + \frac 1 {s^{1/2}} \left(
\begin{array}{cc}
1 & 0 \\
 \alpha i  & 1 \\
 \end{array}
 \right) R_1(z)\left(
  \begin{array}{cc}
  1 & 0 \\
  -\alpha i  & 1 \\
  \end{array}
  \right)+O\left(s^{-1}\right),
\end{equation}
uniformly for $z \in \mathbb{C} \setminus \Sigma_R$ on compact subsets.

Next, a combination of \eqref{R-expand} and RH problem \ref{rhp:R} shows that $R_1(z)$ is a solution of the following RH problem.
\begin{rhp}
The $2\times 2$ matrix-valued function $R_1(z)$ appearing in \eqref{R-expand} has the following properties:
\begin{enumerate}
\item[\rm (1)]  $R_1(z)$ is analytic in $\mathbb{C}\setminus \partial D_{-1}$.

\item[\rm (2)]   For $z \in \partial D_{-1}$, we have
$$R_{1,+}(z)-R_{1,-}(z)=J_1(z),$$
where $J_1(z)$ is given in \eqref{Jump-J1}.
\item[\rm (3)] As $z \to \infty$, $R_1(z)=O(1/z)$.
\end{enumerate}
\end{rhp}
By Cauchy theorem and the residue theorem, we obtain that
\begin{align}\label{eq:R1}
  R_1(z) = \frac{1}{2\pi i} \oint_{\partial D_{-1}} \frac{J_1(\zeta)}{\zeta-z} d\zeta
=\left\{
    \begin{array}{ll}
      \left(
                                                                       \begin{array}{cc}
                                                                         0 & \frac {1}{8i(z+1)  \left(1 - \frac {\lambda^k \eta_0}{s^{k+1/2}} \right)  } \\
                                                                         0 & 0 \\
                                                                       \end{array}
                                                                     \right) - J_1(z), & \hbox{for $z\in D_{-1}$,} \\
      \left(
                                                                       \begin{array}{cc}
                                                                         0 & \frac {1}{8i(z+1) \left(1 - \frac {\lambda^k \eta_0}{s^{k+1/2}} \right) } \\
                                                                         0 & 0 \\
                                                                       \end{array}
                                                                     \right), & \hbox{for $z\in\mathbb{C}\setminus D_{-1}.$}
    \end{array}
  \right.
\end{align}
Furthermore, we note that
\begin{align}\label{R-1 at zero}
  (R_1(z))_{21} &= \frac{i((\rho(z)+\rho(z)^{-1})^2-1) \sqrt{z+1}}{8 g(z) } \nonumber \\
   &=  \frac{3is^{k+1/2}}{8(s^{k+1/2} - \lambda^k \eta_0 )} + i \left[ \frac{\alpha^2}{2} \frac{s^{k+1/2}}{s^{k+1/2} - \lambda^k \eta_0} + \frac{3\eta_1}{8} \frac{ \lambda^k s^{k+1/2}}{(s^{k+1/2} - \lambda^k \eta_0)^2}  \right] (z+1)
   \nonumber
   \\  &~~~ + O(z+1)^2,
\end{align}
as $z \to -1$, where the constants $\eta_0$ and $\eta_1$ are given in \eqref{cont-eta0}.

\subsection{Large $s$ asymptotics of $\frac{\partial}{\partial s} F(s;\lambda)$}

Tracing back the transformations $S \mapsto T \mapsto X$, we have from \eqref{eq:derivativeinX}, \eqref{eq:XtoT} and \eqref{eq:TtoS} that
\begin{align}
  \frac{\partial}{\partial s} F(s;\lambda) & =-
\frac{e^{\alpha \pi i}}{2 \pi i} \lim_{z \to -s} \left(X^{-1}(z)X'(z)\right)_{21} = -
\frac{e^{\alpha \pi i}}{2 \pi is } \lim_{z \to -1} \left(T^{-1}(z)T'(z)\right)_{21} \nonumber \\
& = - \frac{e^{\alpha \pi i}}{2 \pi is } \lim_{z \to -1} \left( e^{-\sqrt{s} g(z) \sigma_3} S^{-1}(z)S'(z)  e^{\sqrt{s} g(z) \sigma_3} \right)_{21}, \label{eq:derivativeinS}
\end{align}
where the derivative is taken with respect to $z$ and the limit as $z \to -s$ is taken as $z \in \textrm{region} \, V$ in Figure \ref{fig:YtoX}.

When $z$ is close to $-1$, it follows from \eqref{def:P-1} and \eqref{def:R} that
\begin{equation}
  S(z) = R(z) E_{-1}(z)\Phi^{(0)}_{\textrm{Bes}}\left(\frac{1}{4}s g^2(z)\right)\left\{
                 \begin{array}{ll}
                   e^{(-\sqrt{s}g(z)+\frac{\alpha \pi i}{2})\sigma_3}, & \hbox{for $\Im z > 0 $,} \\
                   e^{(-\sqrt{s}g(z)-\frac{\alpha \pi i}{2})\sigma_3}, & \hbox{for $\Im z < 0 $.}
                 \end{array}
               \right.
\end{equation}
Thus, if $\Im z < 0$, we obtain
\begin{align*}
  & e^{-\sqrt{s} g(z) \sigma_3} S^{-1}(z)S'(z)  e^{\sqrt{s} g(z) \sigma_3} \\
  & = e^{\frac{\alpha \pi i}{2}\sigma_3} \biggl[ {\Phi^{(0)}_{\textrm{Bes}}\left(\frac{s g^2(z)}{4}\right)}^{-1} E^{-1}_{-1}(z) R^{-1}(z) R'(z) E_{-1}(z) \Phi^{(0)}_{\textrm{Bes}}\left(\frac{s g^2(z)}{4}\right)
\\
&\qquad \qquad \qquad + {\Phi^{(0)}_{\textrm{Bes}}\left(\frac{s g^2(z)}{4}\right)}^{-1} E_{-1}^{-1}(z) E_{-1}'(z) \Phi^{(0)}_{\textrm{Bes}}\left(\frac{s g^2(z)}{4}\right)
\\
&\qquad \qquad \qquad
+ {\Phi^{(0)}_{\textrm{Bes}}\left(\frac{s g^2(z)}{4}\right)}^{-1} \left[\Phi^{(0)}_{\textrm{Bes}} \left(\frac{sg^2(z)}{4}\right)\right]' -\sqrt{s} g'(z) \sigma_3 \biggr] e^{-\frac{\alpha \pi i}{2}\sigma_3}.
\end{align*}
Plugging the above expression into \eqref{eq:derivativeinS} implies that
\begin{align}\label{eq:sderivativefinal}
  &\frac{\partial}{\partial s} F(s;\lambda)
\nonumber
\\
& = \frac{-1}{2 \pi is } \lim_{z \to -1}
  \biggl[ \Phi^{(0)}_{\textrm{Bes}}\left(\frac{s g^2(z)}{4}\right)^{-1} E^{-1}_{-1}(z) R^{-1}(z) R'(z) E_{-1}(z) \Phi^{(0)}_{\textrm{Bes}}\left(\frac{s g^2(z)}{4}\right)
\nonumber
\\
&\qquad \qquad \qquad + \Phi^{(0)}_{\textrm{Bes}}\left(\frac{s g^2(z)}{4}\right)^{-1} E_{-1}^{-1}(z) E_{-1}'(z) \Phi^{(0)}_{\textrm{Bes}}\left(\frac{s g^2(z)}{4}\right)
\nonumber
\\
&\qquad \qquad \qquad
+ \Phi^{(0)}_{\textrm{Bes}}\left(\frac{s g^2(z)}{4}\right)^{-1} \left[\Phi^{(0)}_{\textrm{Bes}} \left(\frac{sg^2(z)}{4}\right)\right]' \; \biggr]_{21}.
\end{align}

We now evaluate the limits of the three terms on the right hand side of \eqref{eq:sderivativefinal} one by one.
We start with the last term. By \eqref{eq:phiBes}, it is readily seen that
\begin{align}\label{eq:term31}
& \biggl[ \Phi^{(0)}_{\textrm{Bes}}\left(\frac{s g^2(z)}{4} \right)^{-1} \left[\Phi^{(0)}_{\textrm{Bes}} \left(\frac{sg^2(z)}{4}\right)\right]' \; \biggr]_{21}
\nonumber \\
& =\left[-2\pi i\left(I_0'\left(2\zeta^{1/2}\right)\right)^2+2 \pi i I_0\left(2\zeta^{1/2}\right)\left(\frac{\zeta^{-1/2}I_0'\left(2\zeta^{1/2}\right)}{2}+I_0''\left(2\zeta^{1/2}\right)\right)\right]\biggr|_{\zeta=s g^2(z)/4}
\nonumber
\\
&~~~\times \left(\frac{sg^2(z)}{4}\right)'.
\end{align}
Recalling the following Taylor expansion as $\xi \to 0$ of the modified Bessel function $I_0$ (cf. \cite{AS}):
$$I_0(\xi)=1+\frac {\xi^2}{4}+\cdots.$$
A direct calculation with the aid of \eqref{def:g}, \eqref{eq:term31} and the above formula shows that
\begin{equation}\label{eq:valueterm3}
  \lim_{z \to -1}
  \biggl[ \Phi^{(0)}_{\textrm{Bes}}\left(\frac{s g^2(z)}{4} \right)^{-1} \left[\Phi^{(0)}_{\textrm{Bes}} \left(\frac{sg^2(z)}{4}\right)\right]' \; \biggr]_{21} = \frac{i \pi }{2s^{2k}} \left( s^{k+1/2} - \lambda^k \eta_0 \right )^2.
\end{equation}
Next, we observe that the $(2,1)$ entry of a matrix $P$ is invariant under the conjugation \newline
$\Phi^{(0)}_{\textrm{Bes}}\left(\frac{s g^2(z)}{4} \right)^{-1} P \Phi^{(0)}_{\textrm{Bes}}\left(\frac{s g^2(z)}{4} \right) $ as $z\to -1$. It is then immediate from \eqref{relation:e-etilde}, \eqref{etilde1} and \eqref{etilde1prime} that
\begin{multline}\label{eq:valueterm2}
  \lim_{z \to -1}
  \biggl[  \Phi^{(0)}_{\textrm{Bes}}\left(\frac{s g^2(z)}{4}\right)^{-1} E_{-1}^{-1}(z) E_{-1}'(z) \Phi^{(0)}_{\textrm{Bes}}\left(\frac{s g^2(z)}{4}\right) \biggr]_{21}
\\ = \left(\widetilde{E}^{-1} (-1) \widetilde{E}'(-1)\right)_{21} = - i \pi \alpha \frac{s^{k+1/2} - \lambda^k  \eta_0 }{s^k}.
\end{multline}
Finally, for the first term, we see from \eqref{relation:e-etilde} and \eqref{R-expand} that
\begin{multline}\label{eq:3.75}
  E_{-1}^{-1} (z)R^{-1}(z) R'(z) E_{-1}(z)
 \\ = \widetilde{E}_{-1}^{-1}(z) \left( I - \frac{R_1(z)}{s^{1/2}}  + O(1/s) \right) \left(\frac{R_1'(z)}{s^{1/2}} + O(1/s) \right) \widetilde{E}_{-1}(z),
\end{multline}
for large positive $s$.
Hence,
\begin{align}\label{eq:valueterm1}
  & \lim_{z \to -1}
  \biggl[ \Phi^{(0)}_{\textrm{Bes}}\left(\frac{s g^2(z)}{4}\right)^{-1} E^{-1}_{-1}(z) R^{-1}(z) R'(z) E_{-1}(z) \Phi^{(0)}_{\textrm{Bes}}\left(\frac{s g^2(z)}{4}\right) \biggr]_{21}
 \nonumber
\\
&=\lim_{z \to -1}
  \left[  E^{-1}_{-1}(z) R^{-1}(z) R'(z) E_{-1}(z)  \right]_{21}=\frac{1}{s^{1/2}}\lim_{z \to -1}
  \left[  \widetilde E^{-1}_{-1}(z) R_1'(z) \widetilde E_{-1}(z)  \right]_{21}+O(1/s)
 \nonumber
\\
& = \pi \frac{s^{k + 1/2} - \lambda^k \eta_0}{s^{k+1/2}} \left(R_1'(-1)\right)_{21} + O(1/s)
=  i \pi \left( \frac{\alpha^2}{2} + \frac{3 \eta_1}{8} \frac{\lambda^k}{s^{k+1/2} - \lambda^k \eta_0} \right)  + O(1/s),
\end{align}
where the first equality follows again from the fact that the $(2,1)$ entry of a matrix $P$ is invariant under the conjugation $\Phi^{(0)}_{\textrm{Bes}}\left(\frac{s g^2(z)}{4} \right)^{-1} P \Phi^{(0)}_{\textrm{Bes}}\left(\frac{s g^2(z)}{4} \right) $ as $z\to -1$, the second equality follows from \eqref{eq:3.75}, and the last two equalities follow directly from \eqref{etilde1} and \eqref{R-1 at zero}. A combination of \eqref{eq:sderivativefinal}, \eqref{eq:valueterm3}, \eqref{eq:valueterm2} and \eqref{eq:valueterm1} then gives us the following large $s$ asymptotics of $\frac{\partial}{\partial s} F(s;\lambda)$:
\begin{align}\label{eq:asysderivative}
 &\frac{\partial}{\partial s} F(s;\lambda)
\nonumber
\\
&=    -\frac{1}{4s^{2k+1}} \left( s^{k+1/2} - \lambda^k \eta_0\right)^2 +  \alpha \frac{s^{k+1/2} - \lambda^k \eta_0}{2s^{k+1}} - \frac{1}{ 2s } \left( \frac{\alpha^2}{2} + \frac{3 \eta_1}{8} \frac{\lambda^k}{s^{k+1/2} - \lambda^k \eta_0} \right)  + O\left(1/s^{2}\right)
\nonumber
\\
&=-\frac{1}{4}+\frac{\alpha}{2s^{1/2}}-\frac{\alpha^2}{4s}+O\left(1/s^{\min(k+1/2,2)}\right).
\end{align}

\section{Large $s$ asymptotics of $\frac{\partial}{\partial \lambda} F(s;\lambda)$ } \label{sec:l-derivative}

To prove Theorem \ref{thm:asy}, it is also necessary to study the large $s$ asymptotics of $\frac{\partial}{\partial \lambda} F(s;\lambda)$, which, together with the large $s$ asymptotics of $\frac{\partial}{\partial s} F(s;\lambda)$, will gives us large $s$ asymptotics of $F(s;\lambda)$ up to a constant term.  It comes out that $\frac{\partial}{\partial \lambda} F(s;\lambda)$ can not be related to the RH problem $X$ directly. Indeed, a straightforward calculation shows that, if
$k=1$,
\begin{align}
& \frac{d}{d\lambda}\widetilde{K}_{\textrm{PIII}}
\nonumber
\\
&=\frac{d}{d\lambda}
\left(\frac{\vec{f}^t(u)\vec{h}(v)}{u-v}\right)
\nonumber
\\
&=-\frac{e^{\alpha \pi i}}{2\pi i uv}\left[q'(\lambda)(\psi_1(u)\psi_2(v)+\psi_2(u)\psi_1(v))-ir'(\lambda)\psi_2(u)\psi_2(v)-
ip'(\lambda)\psi_1(u)\psi_1(v)\right],
\end{align}
where the functions $p$, $q$ and $r$ are the same as those given in \eqref{eq:Psi-infty}. The prefactor $1/(uv)$ in the above formula will cause the difficulty.

To overcome this difficulty, we will consider a `scaled' version of the Fredholm determinant. To this end, we first follow \cite{ACM} and define
\begin{equation}\label{def:phi}
\Phi(z;\lambda)=
\begin{pmatrix}
1 & 0 \\
\lambda^{-1}r\left(\lambda^2\right) & 1
\end{pmatrix}
\lambda^{\frac{\sigma_3}{2}}e^{\frac{\pi i}{4}\sigma_3}\Psi\left(\lambda^2z; \lambda^2\right),
\end{equation}
where $r$ is defined by \eqref{eq:Psi-infty}. Next, for $x<0$, we set
\begin{equation}\label{eq:hatfh}
\widehat{\vec{f}}(x)=\Phi_-(x)
\begin{pmatrix}
-e^{\frac{\alpha \pi i}{2}} \\
e^{-\frac{\alpha \pi i}{2}}
\end{pmatrix}=\begin{pmatrix}
\widehat{f}_1(x) \\
\widehat{f}_2(x)
\end{pmatrix}, \qquad
\widehat{\vec{h}}(x)=-\frac{1}{2 \pi i}
\Phi_-^{-t}(x)
\begin{pmatrix}
e^{-\frac{\alpha \pi i}{2}} \\
e^{\frac{\alpha \pi i}{2}}
\end{pmatrix}=\begin{pmatrix}
\widehat{h}_1(x) \\
\widehat{h}_2(x)
\end{pmatrix}.
\end{equation}
Finally, let us denote by $\widehat{\mathcal{K}}_{\textrm{PIII}}$ the integral operator with kernel
\begin{equation}\label{eq:hatPIII}
\chi_{[-s,0]}(u) \widehat{K}_{\textrm{PIII}}(u,v)\chi_{[-s,0]}(v)=\chi_{[-s,0]}(u) \frac{\widehat{\vec{f}}^t(u)\widehat{\vec{h}}(v)}{u-v}\chi_{[-s,0]}(v)
\end{equation}
acting on the function space $L^2((-\infty,0))$. By \eqref{def:phi} and \eqref{eq:hatfh}, it is readily seen that
\begin{equation}
\chi_{[-s,0]}(u) \widehat{K}_{\textrm{PIII}}(u,v)\chi_{[-s,0]}(v)=\chi_{[-s,0]}(u) \frac{\vec{f}^t(\lambda^2 u;\lambda^2)\vec{h}(\lambda^2v;\lambda^2)}{u-v}\chi_{[-s,0]}(v).
\end{equation}
Thus, a change of variable shows that
$$\ln \det\left(I-\widehat{\mathcal{K}}_{\textrm{PIII}}\right)=F\left(\lambda^2s; \lambda^2\right). $$

Note that
\begin{equation}\label{eq:lambdaderivative}
\frac{d}{d\lambda} \ln \det\left(I-\widehat{\mathcal{K}}_{\textrm{PIII}}\right)=\frac{\partial}{\partial \lambda}F\left(\lambda^2s;\lambda^2\right)
=2\left(F_s\left(\lambda^2s;\lambda^2\right)\lambda s+F_{\lambda}\left(\lambda^2s;\lambda^2\right)\lambda\right).
\end{equation}
Our strategy is the following. For large $s$, on one hand, as we will show later, the asymptotics of  $\frac{d}{d\lambda}\ln \det\left(I-\widehat{\mathcal{K}}_{\textrm{PIII}}\right)$ can be obtained by performing a steepest descent analysis on a RH problem similar to $X$. On the other hand, the asymptotics of $F_s\left(\lambda^2s;\lambda^2\right)$ is available in view of \eqref{eq:asysderivative}. The large $s$ asymptotics of $\frac{\partial}{\partial \lambda}F(s;\lambda)$ can then be derived from \eqref{eq:lambdaderivative} via change of variables $\lambda^2s\rightarrow s$ and $\lambda^2 \rightarrow \lambda$.

In what follows, we first relate $\frac{d}{d\lambda} \ln \det\left(I-\widehat{\mathcal{K}}_{\textrm{PIII}}\right)$ to a RH problem following the same idea laid out in Section \ref{sec:sderivative}.

\subsection{A Riemann-Hilbert setting for $\frac{\partial}{\partial \lambda}F\left(\lambda^2s;\lambda^2\right)$}
To proceed, we note from \eqref{def:phi} and the RH problem for $\Psi$ that $\Phi$ satisfies the following RH problem (see also \cite{ACM}).
\begin{rhp} \label{rhp:Phi}
The function $\Phi(z)=\Phi(z;\lambda)$ defined in \eqref{def:phi} has the following properties:
\begin{enumerate}
\item[\rm (1)] $\Phi(z)$ is defined and analytic in $\mathbb{C}\setminus \{\cup^3_{j=1}\Sigma_j\cup\{0\}\}$, where the contours $\Sigma_j$   are illustrated in Figure \ref{fig:jumpsPsi}.
\item[\rm (2)] $\Phi$ satisfies the following jump conditions:
\begin{equation}\label{eq:Phi-jump}
 \Phi_+(z)=\Phi_-(z)
 \left\{
 \begin{array}{ll}
    \left(
                               \begin{array}{cc}
                                 1 & 0 \\
                               e^{\alpha \pi i}& 1 \\
                                 \end{array}
                             \right), &  z \in \Sigma_1, \\[.4cm]
    \left(
                               \begin{array}{cc}
                                0 & 1\\
                               -1 & 0 \\
                                 \end{array}
                             \right),  &  z \in \Sigma_2, \\[.4cm]
    \left(
                               \begin{array}{cc}
                                 1 & 0 \\
                                 e^{-\alpha \pi i} & 1 \\
                                 \end{array}
                             \right), &   z \in \Sigma_3.
 \end{array}  \right .  \end{equation}

\item[\rm (3)] As $z\to\infty$,
\begin{multline}\label{eq:Phi-infty}
 \Phi(z;\lambda)=
 \begin{pmatrix}
 1 & 0
 \\
 v(\lambda) & 1
 \end{pmatrix}
 \left( I + \frac{1}{z}
 \begin{pmatrix}
 w(\lambda) & v(\lambda) \\
 l(\lambda) & -w(\lambda)
 \end{pmatrix} +
O\left(z^{-2}\right) \right)\\
\times  e^{\frac{\pi i}{4}\sigma_3}z^{-\frac{1}{4} \sigma_3} \frac{I + i \sigma_1}{\sqrt{2}} e^{\lambda z^{\frac{1}{2}} \sigma_3},
   \end{multline}
where
\begin{equation}\label{def:r}
w(\lambda)=q(\lambda^2)/\lambda^2,\quad v(\lambda)=r(\lambda^2)/\lambda, \quad l(\lambda)=p(\lambda^2)/\lambda^3.
\end{equation}
\item[\rm (4)]As $z\to 0$, there exists a matrix $\Phi_0(\lambda)$, independent of $z$, such that
\begin{equation}\label{eq:Phi-origin}
\Phi(z;\lambda)=\Phi_0(\lambda)(I+O(z)) e^{-\left(-\frac{1}{z}\right)^k\sigma_3}z^{\frac \alpha2\sigma_3}\left\{
  \begin{array}{ll}
    I, &  z \in \Omega_1, \\
   \left(
                               \begin{array}{cc}
                                 1 & 0 \\
                               -e^{\alpha\pi i}& 1 \\
                                 \end{array}
                             \right),  &  z \in \Omega_2, \\
    \left(
                               \begin{array}{cc}
                                 1 &0 \\
                                  e^{-\alpha \pi i} &1 \\
                                 \end{array}
                             \right), &   z \in \Omega_3,
 \end{array}  \right .
  \end{equation}
where the regions $\Omega_i$, $i=1,2,3$, are depicted in Figure \ref{fig:jumpsPsi}.
\end{enumerate}
\end{rhp}
Furthermore, by \cite[Equations (2.11) and (2.13)]{ACM}, one has
\begin{equation}
\partial_{\lambda}\Phi
=
\begin{pmatrix}
0 & 1
\\
z-u(\lambda) & 0
\end{pmatrix}\Phi,
\end{equation}
where
\begin{equation}
u(\lambda)=2w(\lambda)-v'(\lambda)-v(\lambda)^2.
\end{equation}

This, together with \eqref{eq:hatfh}, implies that
\begin{equation}
\frac{\partial \widehat{\vec{f}}}{\partial \lambda}(x)
=
\begin{pmatrix}
0 & 1 \\
x-u(\lambda) & 0
\end{pmatrix}
\widehat{\vec{f}}(x),
\qquad
\frac{\partial \widehat{\vec{h}}}{\partial \lambda}(x) = -\begin{pmatrix}
0 &  x-u(\lambda)\\
1 & 0
\end{pmatrix}
\widehat{\vec{h}}(x),
\end{equation}
and by \eqref{eq:hatPIII},
\begin{equation}
\frac{d}{d\lambda}\widehat{K}_{\textrm{PIII}}=\widehat{\vec{f}}^t(u)
\begin{pmatrix}
0 & 1 \\
0 & 0
\end{pmatrix}\widehat{\vec{h}}(v)=\widehat{f}_1(u)\widehat{h}_2(v).
\end{equation}
Thus, we obtain
\begin{equation}\label{eq:lamdaderivative}
\frac{d}{d\lambda} \ln \det\left(I-\widehat{\mathcal{K}}_{\textrm{PIII}}\right)=-\textrm{tr}\left(\left(I-\widehat{\mathcal{K}}_{\textrm{PIII}}\right)^{-1}
\frac{d}{d\lambda}\widehat{\mathcal{K}}_{\textrm{PIII}}\right)
=-\int_{-s}^0 \widehat{F}_1(v)\widehat{h}_2(v)dv,
\end{equation}
where, as in \eqref{def:FH}, $\widehat{\vec{F}}$ and $\widehat{\vec{H}}$ are defined as
\begin{equation}\label{def:FHhat}
\widehat{\vec{F}}=
\begin{pmatrix}
\widehat{F}_1 \\
\widehat{F}_2
\end{pmatrix}:=\left(I-\widehat{\mathcal{K}}_{\textrm{PIII}}\right)^{-1}\widehat{\vec{f}}, \qquad
\widehat{\vec{H}}=
\begin{pmatrix}
\widehat{H}_1 \\
\widehat{H}_2
\end{pmatrix}:=\left(I-\widehat{\mathcal{K}}_{\textrm{PIII}}\right)^{-1}\widehat{\vec{h}}.
\end{equation}

To this end, we also set
$$\widehat{Y}(z)=I-\int_{-s}^0\frac{\widehat{\vec{F}}(w)\widehat{\vec{h}}^t(w)}{w-z}dw.$$
It is then easily seen that
\begin{equation}\label{eq:lamdaderivativeformula1}
\frac{d}{d\lambda} \ln \det\left(I-\widehat{\mathcal{K}}_{\textrm{PIII}}\right)=-\left(\widehat{Y}_\infty \right)_{12},
\end{equation}
where $\widehat{Y}_\infty$ is the residue of $\widehat{Y}$ at infinity, i.e.,
\begin{equation}
\widehat{Y}(z)=I+\frac{\widehat{Y}_\infty}{z}+O\left(z^{-2}\right), \qquad z\to\infty.
\end{equation}

Note that $\widehat{Y}$ also satisfies a RH problem similar to RH problem \ref{rhp:Y}, and the only difference is the jump condition is replaced by
\begin{equation}\label{eq:hatY-jump}
\widehat{Y}_+(x)=\widehat{Y}_-(x)\left(I-2\pi i \widehat{\vec{f}}(x)\widehat{\vec{h}}^t(x)\right),\quad x\in(-s,0).
 \end{equation}
Thus, similar to the definition of $X$ given in \eqref{eq:YtoX}, we
define
\begin{equation}\label{eq:hatYtohatX}
\widehat{X}(z)=\left\{
       \begin{array}{ll}
         \widehat{Y}(z)\Phi(z), & \hbox{for $z$ in region $\textrm{I}\cup \textrm{III}\cup \textrm{IV}$,} \\
         \widehat{Y}(z)\Phi(z)
\begin{pmatrix}
1 & 0 \\
e^{\alpha \pi i} & 1
\end{pmatrix}, & \hbox{for $z$ in region $\textrm{II}$,} \\
         \widehat{Y}(z)\Phi(z)
\begin{pmatrix}
1 & 0 \\
-e^{-\alpha \pi i} & 1
\end{pmatrix}, & \hbox{for $z$ in region $\textrm{V}$,}
       \end{array}
     \right.
\end{equation}
where the regions $\textrm{I}-\textrm{V}$ are shown in Figure \ref{fig:YtoX}. The following RH problem is straightforward to check from the above definition.

\begin{rhp}
The function $\widehat{X}$ defined in \eqref{eq:hatYtohatX} has the following properties:
\begin{enumerate}
\item[\rm (1)] $\widehat{X}(z)$ is defined and analytic in $\mathbb{C}\setminus \{\cup^4_{j=1}\Gamma_j^{(s)}\cup\{-s\}\}$, where $\Gamma_j^{(s)}$ is defined in \eqref{def:gammais}.

\item[\rm (2)] $\widehat{X}$ satisfies the same jump conditions \eqref{eq:X-jump} as $X$.

\item[\rm (3)] As $z\to\infty$,
\begin{equation}\label{eq:hatX-infty}
 \widehat{X}(z)=
 \begin{pmatrix}
 1 & 0 \\
 v(\lambda) & 1
 \end{pmatrix}
 \left( I + \frac{\widehat{X}_\infty}{z}+O\left(\frac{1}{z^2}\right)
 \right) e^{\frac{\pi i}{4}\sigma_3}z^{-\frac{1}{4} \sigma_3} \frac{I + i \sigma_1}{\sqrt{2}} e^{\lambda z^{\frac{1}{2}} \sigma_3},
   \end{equation}
where
\begin{align}\label{def:hatXinfty}
\widehat{X}_\infty & =
\begin{pmatrix}
1 & 0 \\
-v(\lambda) & 1
\end{pmatrix}
\widehat {Y}_\infty
\begin{pmatrix}
1 & 0 \\
v(\lambda) & 1
\end{pmatrix}
+\begin{pmatrix}
 w(\lambda) & v(\lambda) \\
 l(\lambda) & -w(\lambda)
 \end{pmatrix}
\nonumber
\\
&=
\begin{pmatrix}
\ast & \left(\widehat Y_\infty \right)_{12}+v(\lambda) \\
\ast & \ast
\end{pmatrix}.
\end{align}

\item[\rm (4)] We have
\begin{equation}
\widehat{X}(z)=\left\{
                 \begin{array}{ll}
                   O(1)(I+O(z)) e^{\frac{(-1)^{k+1}}{z^k} \sigma_3}z^{\frac \alpha2\sigma_3}, & \hbox{as $z \to 0$,} \\
                   O(\ln(z+s)), & \hbox{as $z \to -s$.}
                 \end{array}
               \right.
\end{equation}

\end{enumerate}
\end{rhp}

The connection between $\widehat{X}$ and $\frac{\partial}{\partial \lambda}F\left(\lambda^2s;\lambda^2\right)$ is stated in the following proposition, which follows directly from
\eqref{eq:lamdaderivativeformula1} and \eqref{def:hatXinfty}.
\begin{prop}
We have
\begin{equation}\label{eq:lamdaderivativef2}
\frac{\partial}{\partial \lambda}F\left(\lambda^2s;\lambda^2\right)=\frac{d}{d\lambda} \ln \det\left(I-\widehat{\mathcal{K}}_{\textrm{PIII}}\right)=v(\lambda)-\left(\widehat{X}_\infty\right)_{12}.
\end{equation}
\end{prop}

\subsection{Asymptotic analysis of the Riemann-Hilbert problem for $\widehat X$}
\label{sec:l-asyanalysis}

In this section, we shall perform the Deift-Zhou steepest descent analysis to the RH problem for $\widehat X$ as $s \to +\infty$, which is similar to the analysis carried out in Section \ref{sec:sasyanalysis}.

\subsubsection{Transformations $\widehat X\to \widehat T \to \widehat S$}
Define consecutive transformations
\begin{equation}\label{eq:hatXtohatT}
\widehat T(z) =\widehat X(sz)
\end{equation}
and
\begin{equation}\label{eq:hatTtohatS}
\widehat S(z)=
s^{\frac{1}{4}\sigma_3}e^{-\frac{1}{4}\pi i \sigma_3}
\begin{pmatrix}
1 & 0
\\
 s \, \widehat{\mathfrak{g}}_1 -v(\lambda)  & 1
\end{pmatrix}\widehat T(z) e^{-\sqrt{s}\widehat g(z)\sigma_3},
\end{equation}
where
\begin{equation}\label{def:hatg}
\widehat g(z):=\sqrt{z+1}\left(\lambda+\frac{(-1)^{k+1}}{s^{k+1/2}} \sum_{j=1}^k \frac{c_j}{z^j} \right), \qquad z\in\mathbb{C}\setminus(-\infty, -1],
\end{equation}
with the coefficients $c_j$ given in \eqref{g-exp-cj},
and
\begin{equation} \label{contant-g1-new}
  \widehat{\mathfrak{g}}_1= \frac{\lambda}{2} + \frac{(-1)^{k+1}}{s^{k+1/2}} c_1.
\end{equation}

From the definition of $\widehat g$ in \eqref{def:hatg}, we have
$$\widehat g_{+}(x)+\widehat g_-(x)=0, \qquad x<-1,$$
and
\begin{equation}
  \widehat g(z)= \left\{
             \begin{array}{ll}
               \lambda \sqrt{z} + \widehat{\mathfrak{g}}_1 z^{-1/2} + O(z^{-3/2}), & \hbox{as $z\to \infty$,} \\
               \frac{(-1)^{k+1}}{s^{k+1/2}z^k} +  O(1), & \hbox{as $z \to 0$,}
             \end{array}
           \right.
\end{equation}
where $\widehat{\mathfrak{g}}_1$ is given in \eqref{contant-g1-new}.


It is then straightforward to check that $\widehat S$ satisfies the following RH problem.
\begin{rhp}
The function $\widehat S$ defined in \eqref{eq:hatTtohatS} has the following properties:
\begin{enumerate}
\item[\rm (1)] $\widehat S(z)$ is defined and analytic in $\mathbb{C} \setminus
\{\cup^4_{j=1}\Gamma_j^{(1)}\cup\{-1\}\}$.

\item[\rm (2)] $\widehat S$ satisfies the following jump conditions:
\begin{equation}\label{eq:hatS-jump}
 \widehat S_+(z)=\widehat S_-(z)
 \left\{
 \begin{array}{ll}
\begin{pmatrix}
e^{\alpha \pi i} & 0 \\
0 & e^{-\alpha \pi i}
\end{pmatrix}, &  z \in \Gamma_1^{(1)}, \\[.4cm]
    \left(
                               \begin{array}{cc}
                                1 & 0\\
                               e^{\alpha \pi i-2\sqrt{s}\widehat g(z)} & 1 \\
                                 \end{array}
                             \right),  &  z \in \Gamma_2^{(1)}, \\[.4cm]
    \left(
                               \begin{array}{cc}
                                0 & 1\\
                               -1 & 0 \\
                                 \end{array}
                             \right),  &  z \in \Gamma_3^{(1)}, \\[.4cm]
    \left(
                               \begin{array}{cc}
                                 1 & 0 \\
                                 e^{-\alpha\pi i-2\sqrt{s}\widehat g(z)} & 1 \\
                                 \end{array}
                             \right), &   z \in \Gamma_4^{(1)}.
 \end{array}  \right .  \end{equation}

\item[\rm (3)] As $z\to\infty$,
\begin{equation}\label{eq:hatS-infty}
 \widehat S(z)=
\left( I + \frac{\widehat{S}_\infty  }{z}+O\left(\frac{1}{z^2}\right) \right)
 z^{-\frac{1}{4} \sigma_3} \frac{I + i \sigma_1}{\sqrt{2}},
   \end{equation}
with
\begin{equation}\label{eq:hatsresidue}
\left(\widehat{S}_\infty \right)_{12}
=     -i\frac{\left(\widehat X_\infty \right)_{12}}{\sqrt{s}}+i\sqrt{s} \, \widehat{\mathfrak{g}}_1,
\end{equation}
where $\widehat{\mathfrak{g}}_1$ is given in \eqref{contant-g1-new}.


\item[\rm (4)]We have
\begin{equation}\label{eq:hatS-origin}
 \widehat{S}(z)=\left\{
        \begin{array}{ll}
          O(1)z^{\frac \alpha2\sigma_3}, & \hbox{as $z\to 0$,} \\
          O(\ln(z+1)), & \hbox{as $z \to -1$.}
        \end{array}
      \right.
 \end{equation}
\end{enumerate}
\end{rhp}

By \eqref{def:hatg}, again, one verifies that for sufficiently large positive $s$,
$$\Re \widehat g(z)>0, \qquad z\in\Gamma_2^{(1)}\cup\Gamma_4^{(1)}.$$
Hence, the jump matrix of $\widehat S$ on $\Gamma_2^{(1)}\cup\Gamma_4^{(1)}$ tends to the identity matrix exponentially fast as $s\to+ \infty$. This uniform convergence is not valid anymore for $z$ near $-1$.

\subsubsection{Outer parametrix, local parametrix near $-1$ and final transformation}
By ignoring the exponentially small terms in the jumps \eqref{eq:hatS-jump} and a neighborhood of $-1$, we are led to the  RH problem \ref{rhp:N} for the outer parametrix $N$, whose solution is given by \eqref{Nsolution}.

In a small, open neighborhood $D_{-1}$ near $-1$ with fixed radius, it is natural to consider a function $\widehat P_{-1}$ satisfying a RH problem as follows.
\begin{rhp} \label{rhp:hatP1}
We look for a $2\times 2$ matrix-valued function
$\widehat P_{-1}(z)$ satisfying
\begin{enumerate}
\item[\rm (1)] $\widehat P_{-1}(z)$ is defined and analytic in  $\overline{D_{-1}} \setminus
\{\cup^4_{j=1}\Gamma_j^{(1)}\cup\{-1\}\}$.
\item[\rm (2)] $\widehat P_{-1}(z)$ satisfies the same jump conditions \eqref{eq:hatS-jump} as $\widehat S$ for $z \in D_{-1}\cap \{\cup^4_{j=1}\Gamma_j^{(1)}\}.$

\item[\rm (3)] As $s \to +\infty$, $\widehat P_{-1}$ matches $N(z)$ on the boundary $\partial D_{-1}$ of $D_{-1}$, i.e.,
\begin{equation}\label{eq:mathcinghatp1}
\widehat P_{-1}(z)=
N(z)(I+o(1)), \qquad z\in \partial D_{-1}.
\end{equation}

\end{enumerate}
\end{rhp}
This RH problem is similar to RH problem \ref{rhp:P1} for $P_{-1}$, the only difference is that the function $g$ in the jump condition now is replaced by $\widehat g$. Hence, following the same arguments in Section \ref{sec:localpara}, it is readily seen that $\widehat P_{-1}(z)$ is constructed in terms of the Bessel parametrix of order 0 as follows:
\begin{equation} \label{def:hatP-1}
\widehat P_{-1}(z)=\widehat E_{-1}(z)\Phi^{(0)}_{\textrm{Bes}}\left(\frac{1}{4}s \widehat g^2(z)\right)\left\{
                 \begin{array}{ll}
                   e^{(-\sqrt{s}\widehat g(z)+\frac{\alpha \pi i}{2})\sigma_3}, & \hbox{for $D_{-1}\cap \{z~|~\Im z > 0 \}$,} \\
                   e^{(-\sqrt{s}\widehat g(z)-\frac{\alpha \pi i}{2})\sigma_3}, & \hbox{for $D_{-1}\cap \{z~|~\Im z < 0 \}$,}
                 \end{array}
               \right.
\end{equation}
where $\Phi^{(0)}_{\textrm{Bes}}$ is given in \eqref{eq:phiBes} with $\alpha=0$ and
\begin{equation}\label{def:hatE}
\widehat E_{-1}(z):=N(z)\left\{\begin{array}{ll}
                   e^{-\frac{\alpha \pi i}{2}\sigma_3}\frac{I - i \sigma_1}{\sqrt{2}}\left(\pi^2 s \widehat g^2(z)\right)^{\frac{1}{4}\sigma_3}, & \hbox{for $\Im z > 0 $,} \\
                   e^{\frac{\alpha \pi i}{2}\sigma_3}\frac{I - i \sigma_1}{\sqrt{2}}\left(\pi^2 s \widehat g^2(z)\right)^{\frac{1}{4}\sigma_3}, & \hbox{for $ \Im z < 0 $.}
                 \end{array}
 \right.
\end{equation}
is analytic in $D_{-1}$.

The final transformation is defined by
\begin{equation}\label{def:hatR}
\widehat R(z)=\left\{
       \begin{array}{ll}
         \widehat S(z)N(z)^{-1}, & \hbox{for $z \in \mathbb{C}\setminus D_{-1} $,} \\
         \widehat S(z) \widehat P_{-1}(z)^{-1}, & \hbox{for $z \in D_{-1}$.}
       \end{array}
     \right.
\end{equation}

We then have the following RH problem for $\widehat R$.
\begin{rhp}\label{rhp:hatR}
The $2\times 2$ matrix-valued function $\widehat R(z)$ defined in \eqref{def:hatR} has the following properties:
\begin{enumerate}
\item[\rm (1)]  $\widehat R(z)$ is analytic in $\mathbb{C} \setminus \Sigma_{R}$,
where the contour $\Sigma_R$ is shown in Figure \ref{fig:ContourR}.

\item[\rm (2)]  $\widehat R(z)$ satisfies the jump condition
$$ \widehat R_+(z)= \widehat R_-(z)J_{\widehat R}(z), \qquad z\in\Sigma_R,$$
where
  \begin{equation}
                     J_{{\widehat R}}(z)=\left\{
                                      \begin{array}{ll}
                                        \widehat P_{-1}(z) N(z)^{-1}, & \hbox{for $z \in \partial D_{-1}$,} \\
                                        N(z) J_{\widehat S}(z) N(z)^{-1}, & \hbox{for $ z \in \Sigma_{R} \setminus \partial D_{-1}$.}
                                      \end{array}
\right.
  \end{equation}
\item[\rm (3)] As $z \to \infty$,
$$\widehat R(z)=I+O(1/z).$$
\end{enumerate}
\end{rhp}
The RH problem \ref{rhp:hatR} is equivalent to the following singular integral equation:
\begin{equation}\label{eq:integraloperator}
\widehat R(z)=I+\frac{1}{2\pi i}\int_{\Sigma_R}\widehat R_-(w)\left(J_{\widehat R}(w)-I\right)\frac{dw}{w-z}, \qquad z\in\mathbb{C}\setminus \Sigma_R.
\end{equation}
Furthermore, for large positive $s$, there exits some constant $c>0$ such that
\begin{equation}\label{eq:hatRexpo}
J_{\widehat {R}}(z)=I+O\left(e^{-c\sqrt{s}}\right), \qquad z\in \Sigma_{R} \setminus \overline{D_{-1}};
\end{equation}
while for $z\in\partial D_{-1}$, as in \eqref{R-jump-estimate}, we have
\begin{align}
J_{\widehat R}(z)
=  I+\frac{1}{s^{1/2}} \begin{pmatrix}
    1 & 0 \\ i \alpha & 1
  \end{pmatrix}  \widehat J_1(z)
  \begin{pmatrix}
    1 & 0 \\ -i \alpha & 1
  \end{pmatrix} + O\left(s^{-1}\right), \label{hatR-jump-estimate}
\end{align}
where
\begin{align}
  \widehat J_{1}(z)=\left(
  \begin{array}{cc}
    \frac {\rho(z)^{-2}-\rho(z)^2}{8 \widehat g(z)} & \frac{-i((\rho(z)+\rho(z)^{-1})^2-3)}{8 \sqrt{z+1} \, \widehat g(z)}  \\
    \frac{-i((\rho(z)+\rho(z)^{-1})^2-1) \sqrt{z+1}}{8 \widehat g(z)} &\frac {\rho(z)^2-\rho(z)^{-2}}{8 \widehat g(z)} \\
  \end{array}
\right)
\end{align}
with $\rho(z)$ defined in \eqref{def:rho}. Hence,
as in \eqref{R-expand}, $\widehat R(z)$ has the following asymptotic expansion
\begin{equation}\label{hatR-expand}
\widehat R(z)=I + \frac 1 {s^{1/2}} \left(
\begin{array}{cc}
1 & 0 \\
 \alpha i  & 1 \\
 \end{array}
 \right) \widehat R_1(z)\left(
  \begin{array}{cc}
  1 & 0 \\
  -\alpha i  & 1 \\
  \end{array}
  \right)+O\left(s^{-1}\right),
\end{equation}
where
\begin{align}\label{eq:hatR1}
  \widehat R_1(z) = \frac{1}{2\pi i} \oint_{\partial D_{-1}} \frac{\widehat J_1(\zeta)}{\zeta-z} d\zeta
=\left\{
    \begin{array}{ll}
      \left(
                                                                       \begin{array}{cc}
                                                                         0 & \frac {1}{8i(z+1)\left(\lambda -  \frac{\eta_0}{s^{k+1/2}  } \right)} \\
                                                                         0 & 0 \\
                                                                       \end{array}
                                                                     \right) - \widehat J_1(z), & \hbox{for $z\in D_{-1}$,} \\
      \left(
                                                                       \begin{array}{cc}
                                                                         0 & \frac {1}{8i(z+1)\left(\lambda -  \frac{\eta_0}{s^{k+1/2}  } \right)} \\
                                                                         0 & 0 \\
                                                                       \end{array}
                                                                     \right), & \hbox{for $z\in\mathbb{C}\setminus D_{-1}$,}
    \end{array}
  \right.
\end{align}
and where $\eta_0$ is given in \eqref{cont-eta0}.

\subsection{Large $s$ asymptotics of $\frac{\partial}{\partial \lambda} F(s;\lambda)$}
By \eqref{eq:lamdaderivativef2}, \eqref{contant-g1-new} and \eqref{eq:hatsresidue}, it follows that
\begin{equation}\label{eq:lamdaderivativef3}
\frac{\partial}{\partial \lambda} F\left(\lambda^2s;\lambda^2\right)=\frac{d}{d\lambda} \ln \det\left(I-\widehat{\mathcal{K}}_{\textrm{PIII}}\right)= v(\lambda)-i\sqrt{s}\left(\widehat S_\infty\right)_{12}- \frac{\lambda}{2}s+O\left(s^{-1/2}\right).
\end{equation}
To find the large $s$ asymptotics of $\left(\widehat S_\infty \right )_{12}$, we first observe from \eqref{def:hatR} that
\begin{equation}\label{eq:hatSinhatR}
\widehat S(z)=\widehat R(z)N(z), \qquad z\in\mathbb{C}\setminus D_{-1}.
\end{equation}
Next, let $z\to \infty$, on one hand, it is readily seen from \eqref{Nsolution} that
\begin{equation}\label{eq:Nasy}
N(z)=\left(I+\frac{N_\infty}{z}+O\left(z^{-2}\right)\right)z^{-\frac{1}{4}\sigma_3}\frac{I+i\sigma_1}{\sqrt{2}},
\end{equation}
where
$$N_{\infty}
=
\begin{pmatrix}
\ast & i\alpha
\\
\ast & \ast
\end{pmatrix}.$$
On the other hand, we see from \eqref{eq:integraloperator} that
\begin{equation}\label{eq:hatRasy}
\widehat R(z)=I+\frac{\widehat R_\infty}{z}+O\left(z^{-2}\right),
\end{equation}
where
$$
\widehat R_\infty=\frac{i}{2 \pi}\int_{\Sigma_R}\widehat R_-(w)\left(J_{\widehat R}(w)-I\right)dw.
$$
If we further let $s\to +\infty$, this, together with \eqref{eq:hatRexpo}, \eqref{hatR-jump-estimate} and \eqref{hatR-expand}, implies that
\begin{align}
\widehat R_\infty&=\frac{i}{2\pi s^{1/2}}\oint_{\partial D_{-1}}\begin{pmatrix}
    1 & 0 \\ i \alpha & 1
  \end{pmatrix}  \widehat J_1(w)
  \begin{pmatrix}
    1 & 0 \\ -i \alpha & 1
  \end{pmatrix}dw + O\left(s^{-1}\right)
\nonumber
\\
&=\frac{1}{ s^{1/2}}
\begin{pmatrix}
\ast & -\frac{i}{8}\frac{1}{\lambda- \frac{\eta_0}{s^{k+1/2}} }
\\
\ast & \ast
\end{pmatrix}+ O\left(s^{-1}\right). \label{eq:hatRasy2}
\end{align}
Finally, inserting \eqref{eq:Nasy}, \eqref{eq:hatRasy} and \eqref{eq:hatRasy2} into \eqref{eq:hatSinhatR}, it is easily seen from \eqref{eq:hatS-infty} that
$$
\left(\widehat S_\infty \right)_{12}=i\alpha-\frac{i}{8s^{1/2}}\frac{1}{\lambda-  \frac{\eta_0}{s^{k+1/2}} }+O(1/s).
$$
Substituting the above asymptotics into \eqref{eq:lamdaderivativef3}, we arrive at
\begin{equation}\label{eq:lamdaderivativef4}
\frac{\partial}{\partial \lambda} F\left(\lambda^2s;\lambda^2\right)= -\frac{\lambda s}{2} +
\alpha s^{1/2} + v(\lambda)-\frac{1}{8\lambda}+O\left(s^{-1/2}\right),\quad s\to+\infty.
\end{equation}

Recall now the identity \eqref{eq:lambdaderivative}. On account of \eqref{eq:asysderivative}, it is readily seen that
\begin{equation}\label{eq:Fsasy}
F_s\left(\lambda^2s;\lambda^2\right)\lambda s=-\frac{\lambda s}{4}+\frac{\alpha}{2}s^{1/2}-\frac{\alpha^2}{4\lambda}+O\left(s^{-1/2}\right).
\end{equation}
A combination of \eqref{eq:lambdaderivative}, \eqref{eq:lamdaderivativef4} and \eqref{eq:Fsasy} gives us
\begin{equation} \label{eq:asylambdaderivative-1}
F_\lambda\left(\lambda^2s;\lambda^2\right) = \frac{v(\lambda)}{2\lambda}+\frac{\alpha^2}{4\lambda^2}-\frac{1}{16\lambda^2}+O\left(s^{-1/2}\right),\quad s\to+\infty,
\end{equation}
or equivalently, by \eqref{def:r},
\begin{equation}\label{eq:asylambdaderivative}
F_\lambda(s;\lambda) = \frac{1}{2\lambda} \left(r(\lambda)+\frac{\alpha^2}{2}-\frac{1}{8}\right)+O\left(s^{-1/2}\right),\quad s\to+\infty.
\end{equation}

We are now ready to prove Theorem \ref{thm:asy}.

\section{Proof of Theorem \ref{thm:asy}}
\label{sec:proofthmasy}

From \eqref{eq:asysderivative} and \eqref{eq:asylambdaderivative}, it is immediate that
\begin{multline}\label{eq:asyPIIIconst}
\ln \det\left(I-\widetilde{\mathcal{K}}_{\textrm{PIII}}\right) = -\frac{1}{4}s+\alpha s^{1/2}-\frac{\alpha^2}{4}\ln s \\
+\int_0^\lambda \frac{1}{2t} \left(r(t)+\frac{\alpha^2}{2}-\frac{1}{8}\right) dt+ \widehat \tau_{\alpha}+ O\left(s^{-1/2}\right),\quad s\to+\infty,
\end{multline}
where $\widehat \tau_\alpha$ is a constant independent of $\lambda$ and $s$. From the asymptotics of $r(\lambda)$ given in \eqref{eq:rsmall}, we have that the integral in the above formula is well-defined, and tends to 0 as $\lambda \to 0^+$.

By \eqref{eq:PIIItoBesl}, we have
\begin{equation}\label{eq:dettransition}
\lim_{\lambda \to 0^+}\ln \det\left(I-\widetilde{\mathcal{K}}_{\textrm{PIII}}\right)=\ln \det\left(I-\mathcal{K}_{\textrm{Bes}}\right),
\end{equation}
where $\mathcal{K}_{\textrm{Bes}}$ is the integral operator with kernel $\chi_{[0,s]}(u) K_{\textrm{Bes}}(u,v) \chi_{[0,s]}(v)$ acting on the function space $L^2((0,\infty))$. We mention that \eqref{eq:dettransition} can also be seen from the consistency of \eqref{eq:TW} with the logarithm of \eqref{eq:TWbessel}, as shown in the Appendix below. Letting $s\to +\infty$ on both side of the above formula, it then follows from the large $s$ asymptotics of $\det(I-\mathcal{K}_{\textrm{Bes}})$ in \eqref{eq:asyBesDet} and \eqref{eq:asyPIIIconst} that
$$\widehat\tau_\alpha= \tau_\alpha = \ln \left(\frac{G(1+\alpha)}{(2\pi )^{\alpha/2}} \right).$$

This completes the proof of Theorem \ref{thm:asy}.
\qed

\section{Derivation of the coupled Painlev\'{e} III system}
\label{sec:lax-pair}

In this section, we derive the coupled Painlev\'{e} III system \eqref{eq:PIIIequ} which is relevant to the gap probability.

\subsection{Riemann-Hilbert problem for $M$}
We start with a matrix-valued function $M(z;\lambda,s)$, which is defined via $X(z;\lambda,s)$ in \eqref{eq:YtoX} in the following way:
\begin{equation}\label{eq:MandX}
M(z;\lambda,s)=
\begin{pmatrix}
1 & 0 \\
a(\lambda;s) & 1
\end{pmatrix}
\lambda^{\frac{\sigma_3}{2}}e^{\frac{\pi i}{4}\sigma_3}X\left(\lambda^2z; \lambda^2,\lambda^2s \right).
\end{equation}
From \eqref{eq:derivativeinX}, it is readily seen that $M$ is related to the gap probability as follows:
\begin{equation}\label{eq:derivative-s-M}
F_s\left(\lambda^2s;\lambda^2\right)=-
\frac{e^{\alpha \pi i}}{2 \pi i \lambda^2} \lim_{z \to -s} \left(M_-^{-1}(z)M_-'(z)\right)_{21}, \quad \lambda>0,
\end{equation}
where the limit is taken as $z$ tends to $-s$ in region $\textrm{V}$. This will be our starting point in the proof of Theorem \ref{thm:TW}.

In view of the RH problem for $X$ stated in Proposition \ref{prop:RHPforX}, we have

\begin{rhp}\label{rhp:M}
The function $M$ defined in \eqref{eq:MandX} has the following properties:
\begin{enumerate}
\item[\rm (1)] $M(z)$ is analytic in $\mathbb{C}\setminus \{\cup^4_{j=1}\Gamma_j^{(s)}\cup\{-s\}\}$, where $\Gamma_j^{(s)}$ is defined in \eqref{def:gammais}, as illustrated in solid lines in Figure \ref{fig:YtoX}.

%
%
%
%
%
%
%
%
%
%
%

\item[\rm (2)] $M(z)$ satisfies the same jump conditions \eqref{eq:X-jump} as $X$.

\item[\rm (3)] As $z\to\infty$,
\begin{equation}\label{eq:M-infty}
 M(z)=
 \begin{pmatrix}
 1 & 0 \\
 a(\lambda;s) & 1
 \end{pmatrix}
 \left( I + \frac{M_\infty}{z}+O\left(\frac{1}{z^2}\right)
 \right) e^{\frac{\pi i}{4}\sigma_3}z^{-\frac{1}{4} \sigma_3} \frac{I + i \sigma_1}{\sqrt{2}} e^{\lambda z^{\frac{1}{2}} \sigma_3},
   \end{equation}
where
\begin{equation}\label{def:a}
a(\lambda;s) = \left(M_\infty \right)_{12}.
\end{equation}

\item[\rm (4)] As $z\to 0$,

\begin{equation}\label{eq:M-0}
M(z)=M^{(0)}(z)e^{\frac{(-1)^{k+1}}{z^k}\sigma_3}z^{\frac{\alpha}{2}\sigma_3} ,
\end{equation}
where $M^{(0)}(z)$ is analytic at $z=0$.
\item[\rm (5)] As $z\to -s$,
\begin{equation}\label{eq:M-s}
M(z)=M^{(s)}(z)\left(I+\frac{1}{2\pi i}\log(z+s)\sigma_+\right)E^{(s)}(z),
\end{equation}
where  $M^{(s)}(z)$ is analytic at $z=-s$, i.e., there exist matrices $M^{(s)}_j$, $j=0,1,2,\ldots,$ such that \begin{equation}\label{eq:expnear0Ms}
M^{(s)}(z)=M^{(s)}_0\left(I+\displaystyle\sum_{j=1}^{\infty}M^{(s)}_j(z+s)^k \right),
\end{equation}
$\sigma_+
=\begin{pmatrix}
0 & 1
\\
0 & 0
\end{pmatrix}$,
and $E^{(s)}(z)$ is a piecewise constant matrix-valued function defined by
\begin{equation}\label{def:E-i}
E^{(s)}(z) = \left\{
       \begin{array}{ll}
        e^{\frac{\alpha\pi i}{2}\sigma_3}, & \hbox{$z\in$ II,} \\
        \left( \begin{array}{cc}
                                                 e^{\frac{\alpha\pi i}{2}} & 0 \\
                                                -  e^{\frac{\alpha\pi i}{2}} &  e^{-\frac{\alpha\pi i}{2}} \\
                                               \end{array}
                                             \right), & \hbox{$z\in$ III,} \\
          \left(
                                               \begin{array}{cc}
                                                e^{-\frac{\alpha\pi i}{2}} & 0 \\
                                             e^{-\frac{\alpha\pi i}{2}} & e^{\frac{\alpha\pi i}{2}} \\
                                               \end{array}
                                             \right), & \hbox{$z\in$ IV} \\
         e^{-\frac{\alpha\pi i}{2}\sigma_3}, & \hbox{$z\in$ V,}
       \end{array}
     \right.
\end{equation}
with the regions II--V being depicted in Figure \ref{fig:YtoX}.
\end{enumerate}
\end{rhp}

The analyticity of the functions $M^{(0)}(z)$ and $M^{(s)}(z)$ follows directly from the jump condition for $M(z)$, and
the definition of the piecewise constant matrix-valued function $E^{(s)}(z)$ in \eqref{def:E-i} (for $M^{(s)}(z)$).

We now show that the RH problem for $M$ is uniquely solvable through a standard vanishing lemma argument.

\begin{lem}[Vanishing Lemma]\label{Vanishing lemma}
Let $\widehat{M}$ satisfies the \lq homogeneous\rq\ version of the RH problem for $M$, i.e., $\widehat{M}$ satisfies RH problem \ref{rhp:M} but with the asymptotic behavior near infinity replaced by
\begin{equation}\label{eq:hat M-infty}
 \widehat{M}(z)=
 O(1/z) e^{\frac{\pi i}{4}\sigma_3}z^{-\frac{1}{4} \sigma_3} \frac{I + i \sigma_1}{\sqrt{2}} e^{\lambda z^{\frac{1}{2}} \sigma_3}.
   \end{equation}
Then, this RH problem has only the trivial solution for the parameters $\alpha>-1$, $s\geq 0$ and $\lambda>0$,
that is, $$\widehat{M}(z) \equiv 0.$$
\end{lem}
\begin{proof}
The case for $s=0$ and  $k=1$ is proved in \cite[Lemma 1]{XDZ}. For general $s>0$ and $k\in \mathbb{N}$, the proof is similar. We omit the details but point out that one important observation in the proof is the complex conjugation relation satisfied by jump matrices \eqref{eq:X-jump}:
$$\overline{\left(J_1^{(s)}\right)_{11}}=\left(J_1^{(s)}\right)_{22},\quad \quad \overline{ \left(J_2^{(s)}\right)_{12}}=\left(J_4^{(s)}\right)_{12},$$
where $J_{i}^{(s)}$, $i=1,2,3,4$, stands for the jump matrix of $M$ restricted on the contour $\Gamma_{i}^{(s)}$.
\end{proof}

By means of the vanishing lemma \ref{Vanishing lemma} \cite{DIZ97,FZ92,Zhou89}, the following theorem is immediate.

\begin{thm}\label{Solvability}
Assume that $\alpha>-1$, $s\geq 0$ and $\lambda>0$. Then, the RH problem \ref{rhp:M} for $M(z)$ is uniquely solvable.
\end{thm}

\subsection{Lax pair equations and a coupled Painlev\'{e} III system}

We next obtain linear differential equations for $M$ with respect to $z, \lambda$ and $s$. This system
of differential equations has a Lax pair form and the compatibility condition of the Lax pair will
then lead to the coupled Painlev\'{e} III system \eqref{coupled PIII}.

\begin{prop}\label{CPIII:Lax pair}
We have the following differential equations for $M=M(z;\lambda,s)$:
\begin{align}\label{Lax pair-M-z}
 \frac{\partial M}{\partial z}&=\left( \sum_{j=1}^{k+1} \frac{A_j}{z^j}+\frac{A_0}{z+s}+\frac{\lambda}{2}\sigma_{-}
     \right)M,
\\
\label{Lax pair-M-lambda}
\frac{\partial M}{\partial \lambda}&=\left(
      \begin{array}{cc}
        0 & 1 \\
        z+2a' & 0 \\
      \end{array}
    \right)M,
\\
\label{Lax pair-M-s}
\frac{\partial M}{\partial s}
&=
\left(\frac{\partial a}{\partial s}\sigma_{-}+\frac{A_0}{z+s}\right)M,
\end{align}
where $\sigma_-=\left(
                             \begin{array}{cc}
                               0 & 0 \\
                               1 & 0 \\
                             \end{array}
                           \right)$, $'=\frac{d}{d\lambda}$,
\begin{align}\label{def:A-i}
A_j=\left(
      \begin{array}{cc}
        -\frac{b_{j}'}{2} &  b_j \\
         -\frac{b_{j}''}{2}+2a'b_j + b_{j+1} & \frac{b_{j}'}{2} \\
      \end{array}
    \right) \qquad \textrm{for } 1\leq j\leq k+1
\end{align}
and
\begin{align} \label{def:A0}
     A_0=\left(
      \begin{array}{cc}
        - \frac{1}{4} + \frac{b_{1}'}{2} &  \frac{\lambda}{2} - b_1  \\
         \frac{b_{1}''}{2}+\left(2a'-s\right)\left(\frac{\lambda}{2} - b_1\right)  & \frac{1}{4} - \frac{b_{1}'}{2} \\
      \end{array}
    \right).
\end{align}
Moreover, the functions $a(\lambda;s),b_1(\lambda;s),\ldots,b_{k+1}(\lambda;s)$ in \eqref{def:A-i} and \eqref{def:A0}
satisfy the coupled Painlev\'{e} III system \eqref{coupled PIII}.
\end{prop}

\begin{proof}
Since the jumps in the RH problem for $M(z)$ are $z,\lambda,s$-independent matrices, we have that the functions
\begin{equation}\label{def:ABC}
A(z;\lambda,s):=\frac{\partial M}{\partial z}M^{-1}, \qquad B(z;\lambda,s):=\frac{\partial M}{\partial \lambda}M^{-1}, \qquad
C(z;\lambda,s):=\frac{\partial M}{\partial s}M^{-1}
\end{equation}
are meromorphic functions with possible isolated singular points at $z=-s$ and $z=0$. From the asymptotic behaviors of $M$ near $z=\infty$, $z=-s$ and $z=0$ as given in \eqref{eq:M-infty}--\eqref{eq:M-s}, it follows that
\begin{align}
  A(z;\lambda,s) & = \sum_{j=1}^{k+1} \frac{A_j}{z^j}+\frac{A_0}{z+s}+\frac{\lambda}{2}\sigma_{-}, \label{eq:Aexp}
\\
  B(z;\lambda,s) & =
    \left(
      \begin{array}{cc}
        0 & 1 \\
        z+ a' +  a^2  - 2 (M_\infty)_{11} & 0 \\
      \end{array}
    \right), \label{M-lamda-M} \\
  C(z;\lambda,s) & = \frac{\partial a}{\partial s}\sigma_{-}+\frac{C_1}{z+s},
\end{align}
for some $z$-independent matrices $A_j$, $j=0,\ldots,k+1$, and $C_1$. By \eqref{eq:expnear0Ms}, we further obtain
\begin{equation}\label{eq:A-0-M-s}
A_0=C_1= \frac{1}{2 \pi i} M^{(s)}_0\sigma_+\left(M^{(s)}_0\right)^{-1},
\end{equation}
where $M^{(s)}_0$ is the leading term of the expansion in \eqref{eq:expnear0Ms}. The regularity of $B(z;\lambda,s)$ in \eqref{M-lamda-M} also implies the $1/z$ term in the expansion of $B(z;\lambda,s)$ as $z \to \infty$ must vanish. Hence, by the definition of $B(z;\lambda,s)$ given in \eqref{def:ABC} and \eqref{eq:M-infty}, we find that
\begin{equation}
  a' - a^2  +2 (M_\infty)_{11} =0.
\end{equation}
This, together with \eqref{M-lamda-M}, gives us \eqref{Lax pair-M-lambda}.

Next, we show that the matrices $A_j$, $j=0,\ldots,k+1$, are given by \eqref{def:A-i} and \eqref{def:A0}. To proceed,
we note that the compatibility condition
$$\frac{\partial^2 M}{\partial z \partial \lambda}=\frac{\partial^2 M}{\partial \lambda \partial z }$$
for the differential equations \eqref{Lax pair-M-z} and \eqref{Lax pair-M-lambda} is the zero curvature
relation
\begin{equation}
  \frac{\partial A}{\partial \lambda} -\frac{ \partial B}{\partial z} +[A,B] = 0,
\end{equation}
where $[L,K] = LK - KL$ stands for the standard commutator of two matrices. Substituting \eqref{eq:Aexp} and \eqref{M-lamda-M} into the above formula, we obtain
\begin{equation}\label{eq:Ajcompat}
  A_{j}' +[A_{j}, B_0] + [A_{j+1}, \sigma_-] = 0,  \qquad  1 \leq j \leq k+1,
\end{equation}
where $B_0 = \left(
      \begin{array}{cc}
        0 & 1 \\
        2a' & 0 \\
      \end{array}
    \right)$ and $A_{k+2} = \mathbf{0}_{2 \times 2}$. Thus, if we set
$$b_j=(A_j)_{12},$$
and rewrite \eqref{eq:Ajcompat} in an entrywise manner, it follows that
\begin{align}
    (A_j)_{11}' - (A_j)_{21} + 2 a' b_j + b_{j+1} &= 0, \label{Aj-ode1} \\
   \lambda b_{j}' + 2 (A_j)_{11} &= 0, \label{Aj-ode2} \\
    (A_j)_{21}' - 4 a' (A_j)_{11} + b_{j+1}' &= 0. \label{Aj-ode3}
\end{align}
Since $\det M \equiv 1$, we have
$$\textrm{tr}A=0.$$
A combination of \eqref{Aj-ode1}, \eqref{Aj-ode2} and the above formula gives us the expression of $A_j$ in \eqref{def:A-i}.

To see the expression of $A_0$ in \eqref{def:A0}, we note that
$$(A_0 +A_1)_{12} = \lambda /2.$$
This follows from inserting \eqref{eq:M-infty} into the definition of $A(z;\lambda,s)$ and comparing the coefficients of $O(1/z)$ term in both sides as $z\to \infty$. Thus, $(A_0)_{12} = \lambda /2 - b_1$. The expression \eqref{def:A0} as well as the relation \eqref{eq:a-b} then follows from the compatibility condition
$$\frac{\partial^2 M}{ \partial \lambda \partial s}=\frac{\partial^2 M}{\partial s \partial \lambda }$$
for the differential equations \eqref{Lax pair-M-lambda} and \eqref{Lax pair-M-s}.

Finally, we derive the coupled Painlev\'{e} system \eqref{coupled PIII} by evaluating the determinant of $A$. By \eqref{eq:A-0-M-s}, it is easily seen that $$\det A_0=0.$$ This, together with \eqref{def:A0}, gives us
the first equation in \eqref{coupled PIII}. Moreover, from the asymptotic behavior of  $M$ near the origin given in \eqref{eq:M-0}, we have
\begin{equation}
  A(z;\lambda,s) = \frac{\partial M^{(0)}}{\partial z} \left(M^{(0)}\right)^{-1} + \left( \frac{(-1)^k k }{ z^{k+1}} + \frac{\alpha}{2z} \right) M^{(0)} \sigma_3 \left(M^{(0)}\right)^{-1},
\end{equation}
which particularly implies that, as $z \to 0$,
\begin{equation}
  \det A(z;\lambda,s) = -\frac{k^2}{z^{2k+2}} + \frac{(-1)^{k+1}\alpha k}{z^{k+2}} + O\left( \frac{1}{z^{k+1}} \right).
\end{equation}
The next $k+1$ equations in \eqref{coupled PIII} then follow from the above formula, \eqref{def:A-i}, \eqref{def:A0} and \eqref{eq:Aexp}.

This completes the proof of Proposition \ref{CPIII:Lax pair}.
\end{proof}

\begin{remark}
From \eqref{Aj-ode1}--\eqref{Aj-ode3}, it follows that each $b_j$ also satisfies a third order ODE. We will not use this fact in this paper.
\end{remark}

\begin{remark} \label{remark-b1}
  When $s=0$, the two singular points $z =0 $ and $z=-s$ of $M(z;\lambda,s)$ coalesce. Then, the coefficient $A_0 $ in \eqref{Lax pair-M-z} and \eqref{Lax pair-M-s} is a zero matrix, which means $b_1 = \lambda /2$.
\end{remark}

\section{Asymptotic analysis of the Riemann-Hilbert problem for $M$}
\label{sec:lambda-small}

In this section, we will perform the Deift-Zhou steepest descent analysis to the RH problem for $M$ as $\lambda \to 0^+$, which will then lead to the asymptotics of $a(\lambda;s)$ and $b_1(\lambda;s)$ in \eqref{eqr: thm-asy} and \eqref{eq:asyb1}, respectively.

\subsection{$M\to U$: Rescaling and normalization}
Define
\begin{equation}\label{eq:M to U}
U(z)=e^{-\frac{\pi i}{4}\sigma_3}\lambda^{-\frac{1}{2}\sigma_3}\left(
           \begin{array}{cc}
             1 & 0 \\
             -a(\lambda;s) & 1 \\
           \end{array}
         \right)
M\left(z/\lambda^2\right).
\end{equation}
Then, it is straightforward to check that $U$ satisfies the following RH problem.
\begin{rhp}
The function $U$ defined in \eqref{eq:M to U} has the following properties:
\begin{enumerate}
\item[\rm (1)] $U(z)$ is defined and analytic in $\mathbb{C} \setminus
\{\cup^4_{j=1}\Gamma_j^{\left(\lambda^2 s \right)}\cup\{-\lambda^2s\}\}$, where the definition of $\Gamma_j^{(s)}$ is given in \eqref{def:gammais}.

\item[\rm (2)] $U$ satisfies the following jump conditions:
\begin{equation}\label{eq:U-jump}
 U_+(z)=U_-(z)
 \left\{
 \begin{array}{ll}
\begin{pmatrix}
e^{\alpha \pi i} & 0 \\
0 & e^{-\alpha \pi i}
\end{pmatrix}, &  z \in \Gamma_1^{\left(\lambda^2 s \right)}, \\[.4cm]
    \left(
                               \begin{array}{cc}
                                1 & 0\\
                               e^{\alpha \pi i} & 1 \\
                                 \end{array}
                             \right),  &  z \in \Gamma_2^{\left(\lambda^2 s \right)}, \\[.4cm]
    \left(
                               \begin{array}{cc}
                                0 & 1\\
                               -1 & 0 \\
                                 \end{array}
                             \right),  &  z \in \Gamma_3^{\left(\lambda^2 s \right)}, \\[.4cm]
    \left(
                               \begin{array}{cc}
                                 1 & 0 \\
                                 e^{-\alpha\pi i} & 1 \\
                                 \end{array}
                             \right), &   z \in \Gamma_4^{\left(\lambda^2 s \right)}.
 \end{array}  \right .  \end{equation}

\item[\rm (3)] As $z\to\infty$,
\begin{equation}\label{eq:U-infty}
 U(z)=
 \left( I + \frac{U_\infty}{z} + O\left(\frac{1}{z^2}\right)
 \right) z^{-\frac{1}{4} \sigma_3} \frac{I + i \sigma_1}{\sqrt{2}} e^{\sqrt{z} \sigma_3}.
   \end{equation}

\item[\rm (4)] As $z\to0$,
\begin{equation} \label{U-0-behavior}
U(z)=O(1)e^{\frac{(-1)^{k+1} \lambda^{2k}}{z^k}\sigma_3}z^{\frac{\alpha}{2}\sigma_3}.
\end{equation}
\item[\rm (5)] As $z\to -\lambda^2s$,
\begin{equation} \label{U-lambda-s-behavior}
U(z)=O(1)\left(I+\frac{1}{2\pi i}\log\left(z+\lambda^2s\right)\sigma_+\right)E^{(\lambda^2 s)}(z),
\end{equation}
where the function $E^{(\lambda^2 s)}(z)$ is a piecewise constant matrix-valued function defined by
\begin{equation}\label{def:Elambda}
E^{(\lambda^2 s)}(z) =
\left\{
  \begin{array}{ll}
        e^{\frac{\alpha\pi i}{2}\sigma_3}, & \hbox{$\arg \left(z+\lambda^2s \right) \in \left( 0,\frac{2\pi}{3}\right)$,} \\
    \left( \begin{array}{cc}
                                                 e^{\frac{\alpha\pi i}{2}} & 0 \\
                                                -  e^{\frac{\alpha\pi i}{2}} &  e^{-\frac{\alpha\pi i}{2}} \\
                                               \end{array}
                                             \right), & \hbox{$\arg \left(z+\lambda^2 s \right) \in \left(\frac{2\pi}{3},\pi \right)$,} \\
     \left(
                                               \begin{array}{cc}
                                                e^{-\frac{\alpha\pi i}{2}} & 0 \\
                                             e^{-\frac{\alpha\pi i}{2}} & e^{\frac{\alpha\pi i}{2}} \\
                                               \end{array}
                                             \right), & \hbox{$\arg \left(z+\lambda^2 s \right) \in \left(-\pi, -\frac{2\pi}{3} \right)$,} \\
     e^{-\frac{\alpha\pi i}{2}\sigma_3}, & \hbox{$\arg \left(z+\lambda^2 s \right) \in \left(-\frac{2\pi}{3},0  \right)$.}
  \end{array}
\right.
\end{equation}
\end{enumerate}
\end{rhp}

\subsection{Outer parametrix}

As $\lambda\to 0^{+}$, the jump contour $\Gamma_1^{\left(\lambda^2 s \right)}$ disappears, it is then natural to consider the following RH problem.
\begin{rhp}\label{rhp:pinfty}
We look for a $2\times 2$ matrix-valued function $\Upsilon(z) $ satisfying
\begin{enumerate}
\item[\rm (1)] $\Upsilon(z)$ is defined and analytic in
$\mathbb{C}\setminus \{\cup^3_{j=1}\Sigma_j\cup\{0\}\}$, where the contours $\Sigma_j$  are illustrated in Figure \ref{fig:jumpsPsi}.

\item[\rm (2)] $\Upsilon(z)$ satisfies the following jump conditions:
\begin{equation}
 \Upsilon_{+}(z)=\Upsilon_{-}(z)
 \left\{
 \begin{array}{ll}
    \left(
                               \begin{array}{cc}
                                1 & 0\\
                               e^{\pi i\alpha} & 1 \\
                                 \end{array}
                             \right),  &  z \in \Sigma_1, \\[.4cm]
    \left(
                               \begin{array}{cc}
                                0 & 1\\
                               -1 & 0 \\
                                 \end{array}
                             \right),  &  z \in \Sigma_2, \\[.4cm]
    \left(
                               \begin{array}{cc}
                                 1 & 0 \\
                                 e^{-\pi i\alpha} & 1 \\
                                 \end{array}
                             \right), &   z \in \Sigma_3.
 \end{array}  \right .  \end{equation}

\item[\rm (3)] As $z \to \infty$, the behavior of $\Upsilon$ is the same as that of $U$ given in \eqref{eq:U-infty}.
\end{enumerate}
\end{rhp}

As shown in \cite[Sec. 5.1]{ACM} and \cite[Sec. 5.1]{XDZ}, the solution to RH problem \ref{rhp:pinfty} can be constructed explicitly with the aid of the Bessel parametrix $\Phi^{(\alpha)}_{\textrm{Bes}}$ of order $\alpha$ \eqref{eq:phiBes} in the following way:
\begin{equation}\label{def:p-out}
\Upsilon(z)= \left( I+\frac{i}{8}\left(4\alpha^2+3\right)\sigma_- \right)\pi^{\frac{\sigma_3}{2}}\Phi^{(\alpha)}_{\textrm{Bes}}(z/4).
\end{equation}
As a consequence, the behavior of $\Upsilon(z)$ as $z\to \infty$ in \eqref{eq:U-infty} can be refined to be
\begin{multline}\label{eq:Phi-Bessel-infty}
\Upsilon(z)=\left[I+\frac {4\alpha^2-1}{128z}\left(
                                           \begin{array}{cc}
                                             4\alpha^2-9 & 16i \\
                                            \frac {i}{12} (4\alpha^2-9)(4\alpha^2-13) & 9-4\alpha^2 \\
                                           \end{array}
                                         \right) +O\left (z^{-\frac 32}\right )\right ]
\\
\times z^{-\frac{1}{4}\sigma_3}\frac{I+i\sigma_1}{\sqrt{2}}
 e^{\sqrt{z}\sigma_3}.
\end{multline}
Moreover, from \eqref{eq:phiBes}, one can verify directly that, as $z\to 0$,
 \begin{equation}\label{def:p-entire}
\Upsilon(z)=\widehat \Upsilon(z)\left\{
                \begin{array}{ll}
                  z^{\frac{\alpha}{2}\sigma_3}\left(I+\frac{\sigma_+}{2 i\sin(\pi\alpha)} \right)\widehat C(z), & \hbox{for $ \alpha \notin \mathbb{Z}$,} \\
                  z^{\frac{\alpha}{2}\sigma_3}\left(I+ \frac{(-1)^{\alpha }}{\pi i}\log\frac{\sqrt{z}}{2}\, \sigma_+ \right) \widehat C(z), & \hbox{for $ \alpha \in \mathbb{Z}$,}
                \end{array}
              \right.
\end{equation}
where $\widehat \Upsilon(z)$ is an entire function and $\widehat C(z)$ is defined by
\begin{equation}\label{def:C}
\widehat C(z) = \begin{cases}
  I, &  \arg z\in\left(-\frac{2}{3}\pi,\frac{2}{3}\pi \right), \\
             \left(
             \begin{array}{cc}
             1 & 0 \\
             -e^{\alpha\pi i} & 1 \\
             \end{array}
             \right), & \arg z\in \left(\frac{2}{3}\pi, \pi \right), \\
\left(
\begin{array}{cc}
1 & 0 \\
e^{-\alpha\pi i} & 1 \\
\end{array}
\right), & \arg z\in \left( - \pi,-\frac{2}{3}\pi \right).
\end{cases}
\end{equation}

\subsection{Local parametrix near the origin}

Near the origin, the outer parametrix $\Upsilon(z)$ is not a good approximation to $U(z)$ due to the singular behaviors of $U(z)$ at $z=0$ and $z=-\lambda^2s$. Thus, we need to construct a local parametrix $P_0$ in a neighborhood of  the origin $D_0:=\{z: |z|<\epsilon\}$ with $\epsilon$ fixed but sufficiently small. 

\begin{rhp}\label{rhp:P0}
We look for a $2\times 2$ matrix-valued function $P_0(z)$ satisfying
\begin{enumerate}
\item[\rm (1)] $P_0(z)$ is defined and analytic in $\overline{D_0}\setminus \{\cup^4_{j=1}\Gamma_j^{\left(\lambda^2 s \right)}\cup\{-\lambda^2s\}\}$.

\item[\rm (2)] $P_0(z)$ satisfies the same jump conditions as $U$ on $ D_0 \cap \{\cup^4_{j=1}\Gamma_j^{\left(\lambda^2 s \right)}\}$.

\item [\rm (3)] $U(z)P_0(z)^{-1}$ is bounded at $z=0$ and $z=-\lambda^2s$.

\item[\rm (4)] As $\lambda \to 0^+$, we have the matching condition
\begin{equation}\label{eq:matchoing U-P}
 P_0(z)=\left(I+O\left(\lambda^{2(1+\alpha)}\right)\right) \Upsilon(z)
\end{equation}
for $z \in \partial D_0 $.
\end{enumerate}
\end{rhp}
The construction of the local parametrix $P_0$ is similar to the ones in \cite{ACM} and \cite{XDZ}.
It is explicitly given by
\begin{equation}\label{def:P-0}
P_0(z)=\widehat\Upsilon(z)\left(I+\left(f(z;\lambda)-h(z)\right)\sigma_+\right)e^{\frac{(-1)^{k+1}\lambda^{2k}}{z^k}\sigma_3}z^{\frac{\alpha}{2}\sigma_3} \widetilde C(z),
\end{equation}
where $\widehat\Upsilon(z)$ is given in \eqref{def:p-entire} and 
\begin{equation}
\widetilde C(z) = \begin{cases}
  I, &  \arg \left(z+ \lambda^2 s \right)\in\left(-\frac{2}{3}\pi,\frac{2}{3}\pi \right ), \\
             \left(
                                                                                                       \begin{array}{cc}
                                                                                                         1 & 0 \\
                                                                                                         -e^{\alpha\pi i} & 1 \\
                                                                                                       \end{array}
                                                                                                     \right), & \arg \left(z+\lambda^2 s\right)\in \left(\frac{2}{3}\pi, \pi \right), \\
\left(
                                                                                                       \begin{array}{cc}
                                                                                                         1 & 0 \\
                                                                                                         e^{-\alpha\pi i} & 1 \\
                                                                                                       \end{array}
                                                                                                     \right), & \arg \left(z+ \lambda^2 s\right)\in \left(- \pi,-\frac{2}{3}\pi \right).
\end{cases}
\end{equation}
The scalar functions $f(z;\lambda)$ and $h(z)$ in \eqref{def:P-0} are defined as follows:
\begin{equation}\label{def:f}
f(z;\lambda)= \frac{e^{-z}}{2\pi i}\int_{-\infty}^{-\lambda^2s} \frac{|x|^{\alpha} e^x}{x-z}  e^{2 \frac{(-1)^{k+1} \lambda^{2k}}{x^{k}}} dx,
\end{equation}
and
\begin{equation} \label{def:h-fun}
  h(z) = f(z; 0) -\begin{cases}
    \displaystyle \frac{z^\alpha}{2i
    \sin(\pi \alpha)}, & \textrm{for } \alpha \notin \mathbb{Z}, \vspace{0.2cm} \\
    \displaystyle \frac{(-1)^\alpha z^\alpha}{\pi i} \log \frac{\sqrt{z}}{2},& \textrm{for } \alpha \in \mathbb{Z}.
  \end{cases}
\end{equation}
It is straightforward to check that $h(z)$ is an entire function, $f(z;\lambda)$ is analytic for $z\in \mathbb{C}\setminus(-\infty,-\lambda^2s)$ and satisfies the jump condition
\begin{equation}\label{eq:f-jump}
f_+(x;\lambda)-f_-(x;\lambda)=|x|^{\alpha} e^{2 \frac{(-1)^{k+1} \lambda^{2k}}{x^{k}}} , \quad x\in(-\infty,-\lambda^2s),
\end{equation}
with the following endpoint behavior
\begin{equation}\label{eq:f-est-1}
f(z;\lambda)=\frac{(\lambda^2s)^{\alpha} e^{-\frac{2}{ s^k}} }{2\pi i} \log\left(z+\lambda^2s\right)+O(1),
\end{equation}
as $z\to-\lambda^2s$.
Moreover, uniformly for $z\in\partial D_0 $, we have
\begin{align}
  f(z;\lambda)-h(z) &= \frac{e^{-z}}{2\pi i} \left[ \int_{-\infty}^0 \frac{|x|^\alpha e^{x}}{x-z} \left( e^{2 \frac{(-1)^{k+1} \lambda^{2k}}{x^{k}}} - 1  \right) dx - \int_{-\lambda^2 s}^0 \frac{|x|^\alpha e^{x}}{x-z} e^{2 \frac{(-1)^{k+1} \lambda^{2k}}{x^{k}}} dx \right] \nonumber \\
  & ~~~ + \frac{z^\alpha}{2i \sin(\pi \alpha)} \nonumber \\
  & =  O\left(\lambda^{2k}\right) + O\left(\lambda^{2(1+\alpha)}\right) + \frac{z^\alpha}{2i \sin(\pi \alpha)} \nonumber \\
   &=  \frac{z^\alpha}{2i \sin(\pi \alpha)} + O\left(\lambda^{2(1+\alpha)}\right), \qquad \textrm{as } \lambda \to 0^+, \textrm{ if } \alpha \notin \mathbb{Z}.   \label{eq:f-est-2}
\end{align}
Similarly, if $\alpha \in \mathbb{Z}$, we also have
\begin{equation} \label{eq:f-est-3}
  f(z;\lambda)-h(z) = \frac{(-1)^\alpha z^\alpha}{\pi i } \log \frac{\sqrt{z}}{2} + O\left(\lambda^{2(1+\alpha)}\right),
\end{equation}
as $\lambda \to 0^+ $, uniformly for $z\in\partial D_0$.

Now let us verify $P_0(z)$ defined in \eqref{def:P-0} is indeed a solution to RH problem \ref{rhp:P0}. By \eqref{def:P-0} and \eqref{eq:f-jump}, it is easily seen that $P_0$ satisfies the same jump conditions as $U$ on $ D_0 \cap \{\cup^4_{j=1}\Gamma_j^{\left(\lambda^2 s \right)}\}$. Regarding item $(3)$ in RH problem \ref{rhp:P0}, it follows from \eqref{U-0-behavior} and \eqref{def:P-0} that as $z \to 0$,
\begin{equation*}
  U(z) P_0(z)^{-1} = O(1) \left(I-\left(f(z;\lambda)-h(z)\right)\sigma_+\right) \widehat \Upsilon(z)^{-1} = O(1),
\end{equation*}
since the functions $f$, $g$ and $\widehat \Upsilon$ are regular near the origin. Similarly, if $z \to - \lambda^2 s $ and $\arg \left(z+\lambda^2 s\right)\in\left(\frac{2 \pi}{3}, \pi \right)$, we have from \eqref{U-lambda-s-behavior} and \eqref{def:P-0} that
\begin{align*}
  U(z) P_0(z)^{-1} &= O(1)\left(I+\frac{1}{2\pi i}\log\left(z+\lambda^2s\right)\sigma_+ \right) \begin{pmatrix}
    e^{\frac{\alpha\pi i}{2}} & 0 \\ -  e^{\frac{\alpha\pi i}{2}} &  e^{-\frac{\alpha\pi i}{2}}
  \end{pmatrix} \begin{pmatrix}
    1 & 0 \\  e^{\alpha\pi i }  &  1
  \end{pmatrix} \\
  & ~~~\times z^{-\frac{\alpha}{2}\sigma_3} e^{-\frac{(-1)^{k+1}\lambda^{2k}}{z^k}\sigma_3} \left(I-\left(f(z;\lambda)-h(z)\right)\sigma_+\right) \widehat \Upsilon(z)^{-1} \\
  & =  O(1) \left(I+\frac{z^\alpha e^{-\alpha \pi i}}{2\pi i} e^{2\frac{(-1)^{k+1}\lambda^{2k}}{z^k}} \log\left(z+\lambda^2s \right)\sigma_+ \right)
\\
&~~~ \times \left(I-\left(f(z;\lambda)-h(z)\right)\sigma_+\right) \widehat \Upsilon(z)^{-1}.
\end{align*}
Note that $h(z)$ and $\widehat \Upsilon(z)$ are entire functions and  $\sigma_+^2=\mathbf{0}_{2\times 2}$, using \eqref{eq:f-est-1}, we obtain
\begin{equation}
  U(z) P_0(z)^{-1} = O(1),
\end{equation}
as required. If $z$ approaches $ - \lambda^2 s$ from other regions, the above equation also holds by similar arguments.

Finally, to see the matching condition on $\partial D_0$, we observe from \eqref{def:p-entire} and \eqref{def:P-0} that, if $\alpha \notin \mathbb{Z}$ and $\lambda \to 0^+$,
\begin{eqnarray*}
  P_0(z) \Upsilon(z)^{-1} =  \widehat \Upsilon(z)\left(I+\left(f(z;\lambda)-h(z)\right)\sigma_+\right) e^{\frac{(-1)^{k+1}\lambda^{2k}}{z^k}\sigma_3} \left(I- \frac{z^\alpha}{2i \sin(\pi \alpha)}\sigma_+\right) \widehat \Upsilon(z)^{-1}.
\end{eqnarray*}
It then follows from \eqref{eq:f-est-2} that \eqref{eq:matchoing U-P} holds for $\alpha \notin \mathbb{Z}$. Similarly, if $\alpha \in \mathbb{Z}$, the claim follows from \eqref{def:p-entire} and \eqref{eq:f-est-3}.


\subsection{Final transformation}

The final transformation is defined by
\begin{equation}\label{def:Q}
Q(z)=\left\{
       \begin{array}{ll}
         U(z)\Upsilon(z)^{-1}, & \hbox{for $z \in \mathbb{C}\setminus D_0 $,} \\
         U(z)P_0(z)^{-1}, & \hbox{for $z \in D_0$.}
       \end{array}
     \right.
\end{equation}
Since the point $z = -\lambda^2 s$ is located inside $D_0:=\{z: |z|<\epsilon\}$ when $\lambda$ is small enough, we may apply the contour deformation such that the jump contours of $U(z)$ and  $\Upsilon(z)$ are the same outside $D_0$. It is then easily seen that $Q$ satisfies the following RH problem.
\begin{rhp}
The function $Q(z)$ defined in \eqref{def:Q} has the following properties:
\begin{enumerate}
\item[\rm (1)]  $Q(z)$ is analytic in $\mathbb{C} \setminus \partial D_0$.

\item[\rm (2)]  $Q(z)$ satisfies the jump condition  $$ Q_+(z)=Q_-(z)J_Q(z),\qquad z \in \partial D_0,$$
where
  \begin{equation}
                     J_{Q}(z)=P_{0}(z) \Upsilon(z)^{-1}.
  \end{equation}
\item[\rm (3)] As $z \to \infty$,
$$Q(z)=I+O(1/z).$$
\end{enumerate}
\end{rhp}

From the matching condition \eqref{eq:matchoing U-P}, we obtain
\begin{equation}
 J_Q(z)=I+O\left(\lambda^{2(1+\alpha)}\right), \qquad \textrm{as } \lambda \to 0^+.
\end{equation}
By standard nonlinear steepest descent arguments (cf. \cite{Deift99book,DZ93}), it then follows that
\begin{equation}\label{eq: Q-est}
 Q(z)=I+O\left(\lambda^{2(1+\alpha)}\right), \qquad \textrm{as } \lambda \to 0^+,
\end{equation}
uniformly for $z$ in the complex plane off the jump contour.

Now we are ready to prove Theorems \ref{thm:cp-solutions} and \ref{thm:TW}.

\section{Proofs of Theorems \ref{thm:cp-solutions} and \ref{thm:TW}}
\label{sec:proofoflastthm}

\subsection{Proofs of Theorem \ref{thm:cp-solutions} and \eqref{eq:asyb1}}

The existence of solutions to the coupled Painelv\'{e} system \eqref{coupled PIII} and their analyticity for $\lambda>0$ follow directly from Theorem \ref{Solvability} and Proposition \ref{CPIII:Lax pair}.

To see the asymptotics of $a(\lambda;s)$ as $\lambda\to 0^+$, we note from \eqref{def:a} and \eqref{eq:M to U} that
\begin{equation}\label{eq:a-U}
a(\lambda;s)=i\lambda^{-1}(U_\infty)_{12}.
\end{equation}
By \eqref{def:Q}, one has
$$U(z)=Q(z)\Upsilon(z), \qquad z\in \mathbb{C}\setminus D_0.$$
The asymptotics of $a(\lambda;s)$ given in \eqref{eqr: thm-asy} then follows from \eqref{eq:Phi-Bessel-infty} and \eqref{eq: Q-est}.

Finally, we show the asymptotics of $b_1(\lambda;s)$ as $\lambda \to 0^+$. In view of \eqref{Lax pair-M-z} and \eqref{eq:M to U}, it is readily seen that
\begin{equation}\label{eq:b-U}
b_1(\lambda;s)=\frac{\lambda}{2} - \lambda i \lim_{z\to-\lambda^2s}\left(z+\lambda^2s\right)\left(U'_{+}(z)U_{+}(z)^{-1}\right)_{12}.
\end{equation}
Since $\lambda \to 0^+$, $z=-\lambda^2 s$ is close to the origin. For $z \in D_0$, we see from
the definitions of $P_0(z)$ in \eqref{def:P-0} and $Q(z)$ in \eqref{def:Q} that
\begin{equation}\label{eq:U exp}
U(z)=Q(z)P_0(z)=Q(z)\widehat \Upsilon(z)\left(I+\left(f(z;\lambda)-h(z)\right)\sigma_+\right)e^{\frac{(-1)^{k+1}\lambda^{2k}}{z^k}\sigma_3} z^{\frac{\alpha}{2}\sigma_3},
\end{equation}
where $f(z;\lambda)$ and $h(z)$ are defined in \eqref{def:f} and \eqref{def:h-fun}, respectively.
By \eqref{eq:U exp}, we further have
\begin{align*}
  U'(z) U(z)^{-1} &=  Q'(z) Q(z)^{-1} + Q(z) \widehat{\Upsilon}'(z) \widehat \Upsilon(z)^{-1} Q(z)^{-1} \\
   & ~~~ + Q(z) \widehat \Upsilon(z) \left(f'(z;\lambda)-h'(z)\right)\sigma_+ \left(I-\left(f(z;\lambda)-h(z)\right)\sigma_+\right) \widehat \Upsilon(z)^{-1} Q(z)^{-1} \\
   & ~~~ + Q(z) \widehat \Upsilon(z) \left(I+\left(f(z;\lambda)-h(z)\right)\sigma_+\right) \left( \frac{(-1)^k k \lambda^{2k}}{z^{k+1}} + \frac{\alpha}{2 z} \right)\sigma_3 \\
   &  \hspace{2cm} \times \left(I-\left(
f(z;\lambda)-h(z)\right)\sigma_+\right) \widehat \Upsilon(z)^{-1} Q(z)^{-1}.
\end{align*}
All functions on the right hand side of the above formula are analytic at $z = - \lambda^2 s$ except $f(z;\lambda)$; see \eqref{eq:f-est-1}. Therefore, as $z\to-\lambda^2s$, we obtain
\begin{align}
  \lim_{z\to-\lambda^2s}\left(z+\lambda^2s\right)\left(U'_{+}(z)U_{+}(z)^{-1}\right)_{12} & = \lim_{z\to-\lambda^2s}\left(z+\lambda^2s\right)f'(z;\lambda) \left( Q(z) \widehat \Upsilon(z) \sigma_+ \widehat \Upsilon(z)^{-1} Q(z)^{-1} \right)_{12} \nonumber \\
  & =  \frac{(\lambda^2s)^{\alpha} e^{-\frac{2}{ s^k}} }{2\pi i} \left(Q\left(-\lambda^2s\right)\widehat \Upsilon\left(-\lambda^2s\right)\right)_{11}^2,
\end{align}
A combination of the above formula and \eqref{eq:b-U} gives us
\begin{align} \label{b1-Q-Upsilon}
b_1(\lambda;s)= \frac{\lambda}{2} - \frac{s^{\alpha} \lambda^{2\alpha+1}}{2\pi }  e^{-\frac{2}{s^k}} \left(Q\left(-\lambda^2s\right)\widehat \Upsilon\left(-\lambda^2s\right)\right)_{11}^2.
\end{align}
From \eqref{def:p-out} and \eqref{def:p-entire}, the entire function $\widehat \Upsilon(z)$ takes the following form:
\begin{align} \label{Upsilon-entire}
  \widehat \Upsilon(z) = \left( I+\frac{i}{8}\left(4\alpha^2+3\right)\sigma_- \right)\pi^{\frac{\sigma_3}{2}}\Phi^{(\alpha)}_{\textrm{Bes}}(z/4) \widehat C(z)^{-1} \begin{pmatrix}
    1 & \ast \\ 0 & 1
  \end{pmatrix} z^{-\frac{\alpha}{2}\sigma_3}.
\end{align}
Note that, in the above formula, the (1,1) entry of $\pi^{\frac{\sigma_3}{2}}\Phi^{(\alpha)}_{\textrm{Bes}}(z/4) \widehat C(z)^{-1}$ remains unchanged when it is multiplied by a lower triangular matrix on its left or by an upper triangular matrix on its right. It then follows from \eqref{eq:phiBes} and \eqref{Upsilon-entire} that
\begin{align*}
  \left(\widehat \Upsilon(z)\right)_{11} = \left(\pi^{\frac{\sigma_3}{2}}\Phi^{(\alpha)}_{\textrm{Bes}}(z/4) \widehat C(z)^{-1}  z^{-\frac{\alpha}{2}\sigma_3} \right)_{11} = \sqrt{\pi} z^{-\alpha/2} I_\alpha\left(z^{1/2}\right) = \frac{\sqrt{\pi}}{2^\alpha} \sum_{k=0}^\infty \frac{(z/4)^k}{k! \Gamma(k+\alpha+1)};
\end{align*}
cf. \cite{AS}. Combining \eqref{eq: Q-est}, \eqref{b1-Q-Upsilon} and the above formula, we obtain
\begin{align}
b_1(\lambda;s)=\frac{\lambda}{2} - \frac{s^{\alpha} \lambda^{2\alpha+1}}{2^{2\alpha+1}  \, \Gamma(\alpha+1)^2}e^{-\frac{2}{s^k}} \left(1+O\left(\lambda^{2(1+\alpha)}\right)\right), \label{eq:b-asy-0}
\end{align}
which is given in \eqref{eq:asyb1}.

This completes the proofs of Theorem \ref{thm:cp-solutions} and \eqref{eq:asyb1}. \qed

\begin{remark}
To prove Theorem \ref{thm:TW}, we need to know a little more about the properties of $M(z)$ as $z \to -s$ in \eqref{eq:M-s} and \eqref{eq:expnear0Ms}. It is easily seen from \eqref{eq:expnear0Ms} that $M_1^{(s)} = \displaystyle\lim_{z \to -s} M^{(s)}(z)^{-1} M^{(s)}(z)'$. Then, it follows from \eqref{eq:M-s}, \eqref{def:E-i} and \eqref{eq:M to U} that
\begin{equation} \label{M1s-relation to Mz}
  \left( M_1^{(s)} \right)_{21} = e^{\alpha \pi i}\lim_{z \to -s} \left(M_-^{-1}(z)M_-'(z)\right)_{21} = e^{\alpha \pi i} \lambda^2 \lim_{z \to -\lambda^2 s} \left(U_-^{-1}(z)U_-'(z)\right)_{21}.
\end{equation}
Using \eqref{eq:U exp} again and conducting similar calculations in the above proofs, we have the following estimate 
\begin{equation}\label{eq:M-1}
\left(M_1^{(s)}\right)_{21}=O\left(\lambda^{2(1+\alpha)}\right),
\qquad \textrm{as } \lambda \to 0^+.
\end{equation}
\end{remark}


\subsection{Proof of Theorem \ref{thm:TW}}
Recall the connection between $F_s\left(\lambda^2s;\lambda^2\right)$ and $M$ given in \eqref{eq:derivative-s-M}, we then have from \eqref{M1s-relation to Mz} that
\begin{equation}\label{eq:derivative-s-M-1}
F_s\left(\lambda^2s;\lambda^2\right)
=- \frac{1}{2 \pi i \lambda^2}\left( M^{(s)}_1\right)_{21}, \qquad \lambda>0.
\end{equation}
Inserting \eqref{eq:M-s} and \eqref{eq:expnear0Ms} into \eqref{Lax pair-M-lambda}, we obtain, on one hand,
\begin{equation}\label{eq:M-0-1}
\frac{d}{d\lambda}\left(M_1^{(s)}\right)_{21}=\left(M^{(s)}_0\right)_{11}^2.
\end{equation}
On the other hand, by \eqref{def:A-i} and \eqref{eq:A-0-M-s}, it follows that
\begin{equation}\label{eq:M-1-b}
\left(M^{(s)}_0\right)_{11}^2=2\pi i \left( \frac{\lambda}{2} - b_1(\lambda;s) \right).
\end{equation}

On account of the fact that (see \eqref{eq:M-1})
$$\lim_{\lambda\to 0^+}\left(M_1^{(s)}\right)_{21}=0,$$
we obtain from \eqref{eq:derivative-s-M-1}--\eqref{eq:M-1-b} that
\begin{equation}
F_s\left(\lambda^2s;\lambda^2\right)=-\frac{1}{\lambda^2}\int_0^{\lambda}\left( \frac{\tau}{2} - b_1(\tau;s) \right)d\tau,
\end{equation}
where the integral is convergent by the asymptotics of $b_1(\lambda;s)$ as $\lambda\to 0^+$ given in \eqref{eq:b-asy-0}.
Equivalently, by \eqref{eq:a-b}, it follows that
\begin{equation}\label{eq:TW-1}
\frac{\partial}{\partial s}F\left(\lambda^2s;\lambda^2\right)=\lambda^2 F_s\left(\lambda^2s;\lambda^2\right)=-\int_0^{\lambda}a_{s}(\tau;s)d\tau.
\end{equation}
Since $F(0;\lambda^2)=0$, we integrate the above formula and arrive at \eqref{eq:TW}.

This completes the proof of Theorem \ref{thm:TW}. \qed

\appendix

\section{Consistency of \eqref{eq:TW} with the logarithm of \eqref{eq:TWbessel}}
From \eqref{eq:TW}, we have
\begin{equation}\label{eq:F-lambda-1}
\frac{\partial}{\partial\lambda}F\left(\lambda^2s;\lambda^2 \right) =-\left(a(\lambda;s)-a(\lambda;0)\right),
\end{equation}
where the function $a(\lambda;s)$ is among the class of special solutions to the coupled Painlev\'{e} III system \eqref{coupled PIII} as stated in Theorem \ref{thm:cp-solutions}. This, together with the second equality in \eqref{eq:lambdaderivative}
and change of variables $\lambda^2s\rightarrow s$ and $\lambda^2\to \lambda$, implies that
\begin{equation}\label{eq:F-a}
sF_s(s;\lambda) =-\frac {1}{2}\sqrt{\lambda}\left(a\left(\sqrt{\lambda};s/\lambda\right)-a\left(\sqrt{\lambda};0\right)\right)-\lambda F_{\lambda}(s;\lambda).
\end{equation}

We next relate the limit $F_s(s;\lambda)$ as $\lambda \to 0^+$ to the RH problem for $X$ in Proposition \ref{prop:RHPforX}.
From \eqref{def:r} and \eqref{eq:asylambdaderivative-1}, we get
\begin{equation}
   F_\lambda\left(\lambda^2s;\lambda^2\right) = \frac{r(\lambda^2)}{2\lambda^2}+\frac{\alpha^2}{4\lambda^2}-\frac{1}{16\lambda^2}+ O\left(s^{-1/2}\right),\quad \textrm{as } s\to+\infty.
\end{equation}
Due to the change of variables $\lambda^2s\rightarrow s$ and $\lambda^2\to \lambda$, the above formula and \eqref{eq:rsmall} gives us
\begin{equation}\label{eq:F-est}
\lambda F_{\lambda}(s;\lambda)=O\left(\lambda^{3/2}\right), \quad \textrm{as } \lambda\to 0^+,
\end{equation}
uniformly for $s$ in any compact subset of $(0,+\infty)$.
From the asymptotics of $a(\lambda;s) $ as $\lambda\to 0^+$ given in \eqref{eqr: thm-asy}, we also have
\begin{equation}
\sqrt{\lambda} a \left(\sqrt{\lambda};0\right)=\frac{1-4\alpha^2}{8}+O\left(\lambda^{1+\alpha}\right).
\end{equation}
Inserting the above two formulas into \eqref{eq:F-a}, we obtain
\begin{equation}\label{eq:F-a-small-lambda1}
F_s(s;\lambda) =-\frac {1}{2s}\left(\sqrt{\lambda}a\left(\sqrt{\lambda};s/\lambda\right)-\frac{1-4\alpha^2}{8}\right)+o(1), \quad \textrm{as } \lambda\to 0^+,
\end{equation}
where $o(1)$ term is uniform for $s$ in any compact subset of $(0,+\infty)$. On the other hand, on account of \eqref{eq:MandX} and \eqref{def:a}, we have
\begin{equation}
\label{eq:a-X}
\sqrt{\lambda}a\left(\sqrt{\lambda};s/\lambda\right)=i\left(X_{\infty}\right)_{12}(\lambda;s),
\end{equation}
where $X_\infty$ is given in \eqref{eq:X-infty}. A combination of \eqref{eq:F-a-small-lambda1} and \eqref{eq:a-X} then gives us
\begin{equation}\label{eq:F-a-small-lambda}
F_s(s;\lambda) =-\frac {i}{2s}\left(X_{\infty}\right)_{12}(\lambda;s)+\frac{1-4\alpha^2}{16s}+o(1), \quad \textrm{as } \lambda\to 0^+,
\end{equation}
uniformly for $s$ in any compact subset of $(0,+\infty)$.

Let us compare our RH problem for $X$ in Proposition \ref{prop:RHPforX} with RH problem 6.2 for the Bessel kernel in Bothner et al. \cite{BIP}.
They satisfy the following simple relation:
\begin{equation}\label{eq:X-BS}
X^{\mathrm{BIP}}(z;s,v=0) =e^{\frac{\pi }{4}i\sigma_3} X(-sz;s,\lambda=0),
\end{equation}
where $X^{\mathrm{BIP}}(z;s,v)$ is the solution to RH problem 6.2 in \cite{BIP}. Similar to \eqref{eq:X-infty}, let us denote the $1/z$ coefficient in the large $z$ expansion of $X^{\mathrm{BIP}}$ by $X^{\mathrm{BIP}}_\infty$. From \cite[Equation (6.3)]{BIP}, the large $z$ expansions of RH problems 6.2 and B.1 in \cite{BIP}, we have
\begin{equation} \label{bip-x-infty}
  \left(X^{\mathrm{BIP}}_\infty\right)_{12} = 2\mathcal{H}_H(\mathfrak{q}, \mathfrak{p}, s) + \frac{4\alpha^2 - 1}{8s},
\end{equation}
where
\begin{equation}
  \mathcal{H}_H(\mathfrak{q}, \mathfrak{p}, s) = \frac{\mathfrak{q}^2-1}{4s} \mathfrak{p}^2 - \frac{\alpha^2 \mathfrak{q}^2}{4s (\mathfrak{q}^2 -1 )} -\frac{\mathfrak{q}^2}{4},
\end{equation}
is the Hamiltonian for the Painlev\'{e} III equation; see \cite[Equation (1.13)]{BIP}. Here, $\mathfrak{q}(s)$ satisfies the equation \eqref{q-pv} and the boundary condition \eqref{q-pv-bc}, and $\mathfrak{p}(s) = 2s \mathfrak{q}'(s)/(\mathfrak{q}^2(s)-1)$.   Then, the formulas \eqref{eq:F-a-small-lambda}-\eqref{bip-x-infty} give us
\begin{equation}\label{eq:F-a-small-lambda-final}
F_s(s;\lambda) =\mathcal{H}_H(s)+o(1), \quad \textrm{as }  \lambda\to 0^+,
\end{equation}
uniformly for $s$ in any compact subset of $(0,+\infty)$.

From \eqref{eq:Psi-origin}, \eqref{eq:psinegx}, one can derive that 
\begin{equation}
   K_{\textrm{PIII}}(u,v;\lambda)=O\left(e^{-c\left[\left(\frac{\lambda}{u}\right)^k + \left(\frac{\lambda}{u}\right)^k \right]} \right), \quad \textrm{as } u,v \to 0^+,
\end{equation}
for $\lambda > 0$ and certain constant $c>0$. Hence, $F_\lambda(s;\lambda)/s$ is bounded when $s$ is small and positive. This, together with \eqref{eq:F-a} and \eqref{eq:F-a-small-lambda-final}, implies that we can interchange the integral with respect to $s$ on both sides of \eqref{eq:F-a-small-lambda-final} and the limit $\lambda \to 0^+$. Therefore, we obtain
\begin{equation}
  \lim_{\lambda\to 0^+} F(s;\lambda) = \int_0^s \mathcal{H}_H(\tau) d\tau.
\end{equation}
Finally, due to the fact that $[-\tau \,\mathcal{H}_H(\tau)]' = \mathfrak{q}^2(\tau)/4$ (cf. \cite[Equation (2.26)]{TW94}), an integration by parts of the above formula gives us the logarithm of \eqref{eq:TWbessel}.

\section*{Acknowledgements}
We thank the anonymous referees for their careful reading and constructive suggestions. We also thank Prof. Yang Chen for drawing our attention to the preprint \cite{Lyu:Gri:Chen2017} upon completion of the present work. Dan Dai was partially supported by grants from the Research Grants Council of the Hong Kong Special Administrative Region, China (Project No. CityU 11300814, CityU 11300115, CityU 11303016). Shuai-Xia Xu was partially supported by National Natural Science Foundation of China under grant numbers 11571376 and 11201493, GuangDong Natural Science Foundation under grant number 2014A030313176. Lun Zhang was partially supported by National Natural Science Foundation of China under grant number 11501120, by The Program for Professor of Special Appointment (Eastern Scholar) at Shanghai Institutions of Higher Learning, and by Grant EZH1411513 from Fudan University.


\end{document}